\documentclass[a4paper,onecolumn,10pt,accepted=2023-11-13]{quantumarticle}
\pdfoutput=1

\usepackage{hyperref} % needed for quantum class
\usepackage{booktabs} % needed for quantum class

\usepackage{amsmath}
\usepackage{amsfonts}

\usepackage[utf8]{inputenc} % enable UTF-8
\usepackage{lipsum} % write pseudo text

\usepackage[english]{babel}

\usepackage[T1]{fontenc}

\usepackage{etoolbox}

\ifdefined\isoverfull
\else
\fi

\usepackage{etoolbox}

\newbool{includeappendix}
\setbool{includeappendix}{true} % default value
\IfFileExists{headers/config/noappendix.config}{
	\setbool{includeappendix}{false}
}{}

\newif\ifincludeappendixx
\ifbool{includeappendix}{
	\includeappendixxtrue
}{
	\includeappendixxfalse
}

\usepackage{xr} % for \externaldocument{...}
\usepackage{filecontents}

\ifbool{includeappendix}{}{
	\input{appendix-labels-loader}

	\externaldocument{appendix-labels}
}

\usepackage{acro} % for \ac

\usepackage{subcaption}

\usepackage{xcolor} % for definecolor

\definecolor{my-full-blue}{HTML}{1F77B4}
\definecolor{my-full-orange}{HTML}{FF7F0E}
\definecolor{my-full-green}{HTML}{2CA02C}
\definecolor{my-full-red}{HTML}{d62728}
\definecolor{my-full-purple}{HTML}{9467bd}
\definecolor{my-full-brown}{HTML}{8c564b}
\definecolor{my-full-pink}{HTML}{e377c2}
\definecolor{my-full-gray}{HTML}{7f7f7f}
\definecolor{my-full-olive}{HTML}{bcbd22}
\definecolor{my-full-cyan}{HTML}{17becf}
\colorlet{my-blue}{my-full-blue!30}
\colorlet{my-orange}{my-full-orange!30}
\colorlet{my-green}{my-full-green!30}
\colorlet{my-red}{my-full-red!30}
\colorlet{my-purple}{my-full-purple!30}

\usepackage{listings}

\usepackage{textcomp}

\usepackage{xcolor}

\usepackage[scaled=0.8]{beramono}

\definecolor{ckeyword}{HTML}{7F0055}
\definecolor{ccomment}{HTML}{3F7F5F}
\definecolor{cstring}{HTML}{2A0099}

\lstdefinestyle{numbers}{
	numbers=left,
	framexleftmargin=20pt,
	numberstyle=\tiny,
	firstnumber=auto,
	numbersep=1em,
	xleftmargin=2em
}

\lstdefinestyle{layout}{
	frame=none,
	captionpos=b,
}

\lstdefinestyle{comment-style}{
	morecomment=[l]//,
	morecomment=[s]{/*}{*/},
	commentstyle={\color{ccomment}\itshape},
}

\lstdefinestyle{string-style}{
	morestring=[b]",%
	morestring=[b]',%
	stringstyle={\color{cstring}},
	showstringspaces=false,%
}

\lstdefinestyle{keyword-style}{
	keywordstyle={\ttfamily\bfseries},
	morekeywords={
		function,
		constructor,
		int,
		bool,
		return,
		returns,
		uint
	},
	morekeywords = [2]{},
	keywordstyle = [2]{\text},
	sensitive=true,
}

\lstdefinestyle{input-encoding}{
	inputencoding=utf8,
	extendedchars=true,
	literate=
	{ℝ}{$\reals$}1%
	{→}{$\rightarrow$}1%
	{α}{$\alpha$}1%
	{β}{$\beta$}1%
	{λ}{$\lambda$}1%
	{θ}{$\theta$}1%
	{ϕ}{$\phi$}1%
}

\lstdefinestyle{escaping}{
	moredelim={**[is][\color{blue}]{\%}{\%}},
	escapechar=|,
	mathescape=true
}

\lstdefinestyle{default-style}{
	basicstyle=\fontencoding{T1}\ttfamily\footnotesize,
	style=numbers,
	style=layout,
	style=comment-style,
	style=string-style,
	style=keyword-style,
	style=input-encoding,
	style=escaping,
	tabsize=2,
	upquote=true
}

\lstdefinelanguage{BASIC}{
	language=C++,
	style=default-style
}[keywords,comments,strings]%

\usepackage[capitalize]{cleveref}

\makeatletter
\newcommand\theHALG@line{\thealgorithm.\arabic{ALG@line}}
\makeatother

\crefformat{section}{\S#2#1#3}

\newcommand{\secref}[1]{\S#1}

\crefrangeformat{section}{\S#3#1#4\crefrangeconjunction\S#5#2#6}

\crefmultiformat{section}{\S#2#1#3}{\crefpairconjunction\S#2#1#3}{\crefmiddleconjunction\S#2#1#3}{\creflastconjunction\S#2#1#3}

\newcommand{\crefrangeconjunction}{--}

\crefname{table}{Tab.}{Tabs.}
\crefname{listing}{Lst.}{listings}
\crefname{line}{Lin.}{Lin.}
\crefname{appendix}{App.}{App.}

\newcommand{\appref}[1]{%
	\ifbool{includeappendix}{\cref{#1}}{the appendix}%
}
\newcommand{\Appref}[1]{%
	\ifbool{includeappendix}{\cref{#1}}{The appendix}%
}

\usepackage{tikz}
\usepackage{pgfplots}
\usepgfplotslibrary{units}

\usepackage{qcircuit}
\usetikzlibrary{arrows}
\usetikzlibrary{arrows.meta}
\usetikzlibrary{automata}
\usetikzlibrary{calc}
\usetikzlibrary{backgrounds}
\usetikzlibrary{decorations.markings}
\usetikzlibrary{decorations.pathmorphing}
\usetikzlibrary{decorations.pathreplacing}
\usetikzlibrary{fit}
\usetikzlibrary{matrix}
\usetikzlibrary{patterns}
\usetikzlibrary{positioning}
\usetikzlibrary{shadows}
\usetikzlibrary{shapes}
\usetikzlibrary{shapes.geometric}

\tikzstyle{split} = [
	rounded corners=5,
	fill=\backgroundcolor,
	draw=\backgroundcolor,
	line width=#1
]
\tikzset{split/.default=-0.5cm}

\usepackage{physics} % for norm, bra, ket
\usepackage{centernot}
\usepackage{mathtools}
\usepackage{stmaryrd}
\usepackage{float}
\usepackage{amsthm}
\usepackage[bb=boondox]{mathalfa} % for mathbbb
\usepackage{xifthen} % \isempty
\usepackage{cancel}
\usepackage{xspace}
\usepackage{enumitem}	% for custom enumeration labels
\usepackage{csvsimple}
\usepackage{etoolbox}
\usepackage{standalone}
\usepackage{aligned-overset}
\usepackage{IEEEtrantools}

\usepackage{amssymb}% http://ctan.org/pkg/amssymb
\usepackage{pifont}% http://ctan.org/pkg/pifont

\usepackage{makecell} % for \makecell

\usepackage{marvosym}
\let\marvosymLightning\Lightning
\renewcommand{\lightning}[0]{\text{\marvosymLightning}}

\newtheorem{theorem}{Theorem}[section]

\newcommand{\xyes}[0]{\ding{51}}

\newenvironment{talign}
 {\align}
 {\endalign}
\newenvironment{talign*}
 {\csname align*\endcsname}
 {\endalign}

\DeclareMathAlphabet\mymathcal{OMS}{cmsy}{m}{n}
\DeclareMathAlphabet\mathbfcal{OMS}{cmsy}{b}{n}

\DeclareFontFamily{U}{yswab}{}
\DeclareFontShape{U}{yswab}{m}{n}{<->yswab}{}
\newcommand{\textswab}[1]{\text{\usefont{U}{yswab}{m}{n}#1}}

\DeclareUrlCommand\UScore{\urlstyle{rm}}
\newcommand{\expUScore}{%
  \expandafter\expandafter\expandafter
  \UScore
  \expandafter\expandafter\expandafter
}

\tikzstyle{signstyle}=[line width=0.3ex, x=1ex, y=1ex]
\newcommand{\Plus}{%
	\mathop{%
		\begin{tikzpicture}
			\draw[signstyle] (0.5,0) -- (0.5,1); % |
			\draw[signstyle] (0,0.5) -- (1,0.5); % -
		\end{tikzpicture}
	}%
}
\newcommand{\Minus}{%
	\mathop{%
		\begin{tikzpicture}
			\draw[signstyle, draw opacity=0] (0.5,0) -- (0.5,1); % |
			\draw[signstyle] (0,0.5) -- (1,0.5); % -
		\end{tikzpicture}
	}%
}
\newcommand{\PlusMinus}{%
	\mathop{%
		\begin{tikzpicture}
			\draw[signstyle] (0.5,0) -- (0.5,1); % |
			\draw[signstyle] (0,0.5) -- (1,0.5); % -
			\draw[signstyle] (0,0) -- (1,0); % _
		\end{tikzpicture}
	}%
}

\newcommand{\bottom}[0]{{\perp}}

\let\oldotimes\otimes
\renewcommand{\otimes}{\mathop{\oldotimes}}

\newcommand{\af}[1]{#1^\#} % abstract function
\newcommand{\ad}[0]{\mathbfcal{X}} % abstract domain
\newcommand{\abe}[1]{\boldsymbol{#1}} % abstract element
\newcommand{\cd}[0]{\mymathcal{X}} % concrete domain

\newcommand{\I}[1][]{\ensuremath{\mathbb{I}_{#1}}}

\newcommand{\applyat}[2][]{{#2_{\ifthenelse{\isempty{#1}}{}{(#1)}}}}

\newcommand{\pauliformat}[2]{\ensuremath{\applyat[#2]{\mathit{#1}}}}
\newcommand{\X}[1][]{\pauliformat{X}{#1}}
\newcommand{\Y}[1][]{\pauliformat{Y}{#1}}
\newcommand{\Z}[1][]{\pauliformat{Z}{#1}}

\newcommand{\Ie}{
	\begin{psmallmatrix}
		0 \\ 0
	\end{psmallmatrix}
}
\newcommand{\Xe}{
	\begin{psmallmatrix}
		1 \\ 0
	\end{psmallmatrix}
}
\newcommand{\Ye}{
	\begin{psmallmatrix}
		1 \\ 1
	\end{psmallmatrix}
}
\newcommand{\Ze}{
	\begin{psmallmatrix}
		0 \\ 1
	\end{psmallmatrix}
}

\newcommand{\Imatrix}{
	\begin{psmallmatrix}
		1 & 0 \\
		0 & 1
	\end{psmallmatrix}
}
\newcommand{\Xmatrix}{
	\begin{psmallmatrix}
		0 & 1 \\
		1 & 0
	\end{psmallmatrix}
}
\newcommand{\Ymatrix}{
	\begin{psmallmatrix}
		0 & -\i \\
		\i & \phantom{-}0
	\end{psmallmatrix}
}
\newcommand{\Zmatrix}{
	\begin{psmallmatrix}
		1 & \phantom{-}0 \\
		0 & -1
	\end{psmallmatrix}
}

\newcommand{\pref}{\textswab{f}} % prefactor
	
\newcommand{\bare}{\textswab{b}}
	
\newcommand{\grppref}{\textswab{F}}
\newcommand{\encpauli}[0]{\textswab{g}}

\newcommand{\commut}[2]{#1 \diamond #2}

\newcommand{\ttrace}[1]{\trace \left( #1 \right)}
\newcommand{\join}[0]{\mathbin{\sqcup}}

\usepackage{scalerel,stackengine}
\newcommand\equalhat{\mathrel{\stackon[1.5pt]{=}{\stretchto{%
    \scalerel*[\widthof{=}]{\wedge}{\rule{1ex}{3ex}}}{0.5ex}}}}

\newcommand{\unitary}[1][]{\mathcal{U}\ifthenelse{\isempty{#1}}{}{(#1)}}

\newcommand{\pauli}[0]{P} % pauli operator
\newcommand{\pauligroup}[1]{\mathcal{P}_{#1}}

\newcommand{\reals}[0]{\mathbb{R}}
\newcommand{\complex}[0]{\mathbb{C}}
\newcommand{\bools}[0]{\mathbb{B}}

\usepackage{contour}
\contourlength{0.01em}
\contournumber{5}
\newcommand{\ccontour}[1]{%
	\colorlet{contour-saved}{.}%
	\contour{contour-saved}{#1}%
}
\newcommand{\mathbbbcontour}[1]{
	\ccontour{\ensuremath{\mathbbb{#1}}}
}
\newcommand{\abools}{\mathbbbcontour{B}} % abstract bools
\newcommand{\acomplex}{\mathbbbcontour{C}} % abstract complex
\newcommand{\areals}{\mathbbbcontour{R}} % abstract complex
\newcommand{\apaulis}[1][n]{\mathbbbcontour{P}_{#1}}
\newcommand{\azz}{\mathbbbcontour{Z}_4}
\newcommand{\adensities}{\mathbbbcontour{D}}

\renewcommand{\i}[0]{{\rm i}}

\newcommand{\inset}[1]{{\{{#1}\}}} % abstract element
\newcommand{\scalmult}{ }

\newcommand{\backgroundcolor}{gray!12}

\makeatletter
\newcommand{\xMapsto}[2][]{\ext@arrow 0599{\Mapstofill@}{#1}{#2}}
\def\Mapstofill@{\arrowfill@{\Mapstochar\Relbar}\Relbar\Rightarrow}
\makeatother

\newcommand{\para}[1]{\paragraph{#1.}}

\newcommand{\tool}[0]{\textsc{Abstraqt}\xspace}
\newcommand{\baseline}[0]{YP21\xspace}

\newcommand{\twofrac}[0]{{\color{gray}\tfrac{1}{2}}}

\begin{document}

% \title{\tool: Generalized Stabilizer Simulation to Abstract Quantum States}
% \title{\tool: Stabilizer Abstraction for Quantum Circuits}
% \title{\tool: Efficient Quantum Circuit Simulation by Over-Abstracting Stabilizer Simulation}
\title{Abstraqt: Analysis of Quantum Circuits via Abstract Stabilizer Simulation}

\author{Benjamin Bichsel}
\affiliation{ETH Zurich, Switzerland}
\author{Anouk Paradis}
\affiliation{ETH Zurich, Switzerland}
\author{Maximilian Baader}
\affiliation{ETH Zurich, Switzerland}
\author{Martin Vechev}
\affiliation{ETH Zurich, Switzerland}
\maketitle

%%%%%%%%
% BODY %
%%%%%%%%

% make sure the abstract is in a separate file (abstract.tex). This file will be
% used to generate a markdown version of your abstract
\begin{abstract}
	% This text will be sanitized and placed into lqa-output/abstract.txt
%
%
Stabilizer simulation can efficiently simulate an important class of quantum
circuits consisting exclusively of Clifford gates. However, all existing
extensions of this simulation to arbitrary quantum circuits including
non-Clifford gates suffer from an exponential runtime.

To address this challenge,  we present a novel approach for
efficient stabilizer simulation on arbitrary quantum circuits, at the cost of
lost precision.
Our key idea is to compress an exponential sum representation of the quantum
state into a single \emph{abstract} summand covering (at least) all occurring
summands.
This allows us to introduce an \emph{abstract stabilizer simulator} that
efficiently manipulates abstract summands by \emph{over-approximating} the
effect of circuit operations including Clifford gates, non-Clifford gates, and
(internal) measurements.

We implemented our abstract simulator in a tool called \tool and experimentally
demonstrate that \tool can establish circuit properties intractable for
existing techniques.

\end{abstract}

%{\input{test}}
\section{Introduction} \label{sec:introduction}
Stabilizer simulation~\cite{gottesmanHeisenberg1998} is a promising technique
for efficient classical simulation of quantum circuits consisting exclusively of
\emph{Clifford} gates.
Unfortunately, generalizing stabilizer simulation to arbitrary circuits
including non-Clifford gates requires exponential
time~\cite{aaronson_improved_2004,gottesmantypes,kissinger_simulating_2022,bravyi_simulation_2019,pashayan_fast_2022,kissinger_classical_2022}.
Specifically, the first such generalization by Aaronson and
Gottesman~\cite[\secref{VII-C}]{aaronson_improved_2004} tracks the quantum state
$\rho$ at any point in the quantum circuit as a sum whose number of summands $m$
grows exponentially with the number of non-Clifford gates:
\begin{align} \label{eq:intro-rho}
	\rho
	=
	\sum\limits_{i=1}^m c_i P_i
		\prod\limits_{j=1}^n
		\tfrac{\I + (-1)^{b_{ij}} Q_j}{2}.
\end{align}
Here, while $c_i$, $P_i$, $b_{ij}$, and $Q_j$ can be represented efficiently
(see \cref{sec:background}), the overall representation is inefficient due to
exponentially large $m$.

\para{Abstraction}
The key idea of \tool is to avoid tracking the exact state $\rho$ of a quantum
system and instead only track key aspects of $\rho$.
To this end, we rely on the established framework of abstract
interpretation~\cite{cousotAbstract1977,cousot1992abstract}, which is
traditionally used to analyze classical
programs~\cite{blanchet_static_2003,logozzo_pentagons_2010} or neural
networks~\cite{gehr_ai2_2018} by describing sets of possible states without
explicitly enumerating all of them.
Here, we use abstract interpretation to describe the set of quantum states that
could occur at a specific point during execution of a circuit, by
\emph{over-approximating} the summands that could occur in any of those quantum
states $\rho$.

\para{Merging Summands}
This allows us to curb the exponential blow-up of stabilizer simulation by
merging multiple summands in \cref{eq:intro-rho} into an abstract single summand which
over-approximates all summands, at the cost of lost precision.
The key technical challenge addressed by our work is designing a suitable
\emph{abstract domain} to describe sets of summands, accompanied by the
corresponding \emph{abstract transformers} to over-approximate the actions
performed by the original exponential stabilizer simulation on individual
summands.

As a result, our approach is both efficient and exact on Clifford circuits, as
these circuits never require merging summands. On non-Clifford circuits, merging
summands trades precision for efficiency.
Moreover, our approach naturally allows us to merge the possible outcomes of a
measurement into a single abstract state, preventing an exponential path
explosion when simulating multiple internal measurements.

\para{Main Contributions}
Our main contributions are:
\begin{itemize}
	\item An abstract domain~(\cref{sec:abstractElements}) to over-approximate a
	quantum state represented by \cref{eq:intro-rho}.
	\item Abstract transformers~(\cref{sec:abstract-transformers}) to simulate
	quantum circuits, including gate applications and measurements.
	\item An efficient implementation\footnote{Our implementation is available
	at \url{https://github.com/eth-sri/abstraqt}.} of our approach in a tool
	called \tool~(\cref{sec:implementation}), together with an evaluation
	showing that \tool can establish circuit properties that are intractable for
	existing tools~(\cref{sec:eval}).
\end{itemize}

\para{Results}
Overall, we find that \tool is useful in scenarios where a full simulation of a given circuit
is intractable, but establishing specific properties of the considered circuit
is desirable.

For example, in our evaluation~(\cref{sec:eval}), we demonstrate that
\tool can establish that a circuit ultimately restores some qubits to state
$\ket{0}$.
As precisely simulating the entire circuit is intractable, \tool is typically
the only existing tool able to establish this fact on $12$ benchmarked circuits.
In contrast, existing tools typically yield incorrect results, throw errors, run
out of memory, time out, or are too imprecise to establish the resulting state
is $\ket{0}$.

\para{Outlook}
\tool trades precision for efficiency by abstracting the stabilizer simulation
from \cref{eq:intro-rho}, therefore allowing to establish properties of quantum
circuit outputs when full simulation is intractable. Such results may be useful
for tasks like (i)~establishing that an internal circuit state allows for
specific optimizations, (ii)~debugging quantum computers by establishing
invariants that can be checked at runtime, and more generally (iii)~static
analysis of quantum circuits, or (iv)~verification of the correctness of quantum
circuits.

Further, as discussed in \cref{sec:limitations}, \tool abstracts the very first
stabilizer simulation generalized to non-Clifford gates by Aaronson and
Gottesman~\cite[\secref{VII-C}]{aaronson_improved_2004}. We believe that our
encouraging results pave the way to introduce analogous abstraction to various
follow-up works which improve upon this
simulation~\cite{kissinger_simulating_2022,bravyi_simulation_2019,pashayan_fast_2022,kissinger_classical_2022}.
As these more recent works scale better than
\cite[\secref{VII-C}]{aaronson_improved_2004}, we expect that a successful
application of abstract interpretation to them will yield even more favorable
trade-offs between precision and efficiency.

\section{Background} \label{sec:background}
We first introduce the necessary mathematical concepts.

\para{Basic Notation}
We use $\mathbb{Z}_n := \mathbb{Z} / (n \mathbb{Z})$, define $\bools :=
\mathbb{Z}_2$, and write $2^S$ for the power set of the set $S$.

We represent a pure $n$-qubit quantum state $\psi \in \complex^{2^n}$ as a
\emph{density matrix} $\rho \in \complex^{2^n \times 2^n}$, defined as $\rho
=\psi \psi^\dagger$, where $\psi^\dagger$ denotes the conjugate transpose of
$\psi$. For a mixed state, i.e., a distribution over pure states $\psi_i$ with
probability $p_i$, the corresponding density matrix is $\rho=\sum_{i} p_i \psi_i
\psi_i^\dagger$.
Because both $\psi$ and $\rho$ store exponentially many values, they cannot be represented explicitly for large $n$.
We denote the embedding of a $k$-qubit gate $U \in \unitary[2^k]$ as an
$n$-qubit gate by $\applyat[i]{U} := \I[2^i] \otimes U \otimes \I[2^{n-i-k}]$,
where $\I[l]$ denotes the $l \times l$ identity matrix.

\para{Stabilizer Simulation}
The key idea of stabilizer
simulation~\cite{gottesmanHeisenberg1998,aaronson_improved_2004} is representing
quantum states $\rho=\psi\psi^\dagger$ implicitly, by \emph{stabilizers} $Q$
which stabilize the state $\psi$, that is $Q \psi = \psi$.
As shown in~\cite{aaronson_improved_2004}, appropriately selecting $n$
stabilizers $Q_j$ then specifies a unique $n$-qubit state $\rho = \prod_{j=1}^n
\tfrac{\I + Q_j}{2}$. 

In stabilizer simulation, all $Q_j$ are \emph{Pauli elements} from
$\pauligroup{n}$ of the form $\i^v \cdot \pauli^{(0)} \otimes
\cdots \otimes \pauli^{(n-1)}$, where $\pauli^{(j)} \in \{X, Y, Z, \I[2]\}$ and v $\in \mathbb{Z}_4$.
This directly implies that all stabilizers $Q_i$ for the same state $\psi$
commute, that is $Q_i Q_j = Q_j Q_i$, as elements from the Pauli group
$\pauligroup{n}$ either commute or anti-commute.
These elements can be represented efficiently in memory by storing $v$ and
$\pauli^{(0)}, \dots, \pauli^{(n-1)}$. 
In \cref{app:stabilizers}, we list states stabilized by Pauli matrices
(\cref{tab:stabilizers}) and the results of multiplying Pauli matrices
(\cref{tab:multiply-pauli}).
Further, in this work we use the functions \emph{bare} $\bare \colon
\pauligroup{n} \to \pauligroup{n}$ and \emph{prefactor} $\pref \colon
\pauligroup{n} \to \mathbb{Z}_4$ which extract the Pauli matrices without the
prefactor and the prefactor, respectively: 
\begin{align}
	\pref(\i^v P^{(0)} \otimes \cdots \otimes P^{(n-1)}) &= v, \label{eq:prefactor} \\
	\bare(\i^v P^{(0)} \otimes \cdots \otimes P^{(n-1)}) &= P^{(0)} \otimes \cdots \otimes P^{(n-1)}. \label{eq:barepauli}
\end{align}

Applying gate $U$ to state $\rho$ can be
reduced to conjugating the stabilizers $Q_j$ with $U$: 
\begin{talign}
	U \rho U^\dagger
	=
	U \Big( \prod\limits_{j=1}^n \tfrac{\I +
	{Q_j}}{2} \Big) U^\dagger
	\overset{\text{\cite[Sec.~10.5]{nielsen_quantum_2010}}}{=}
	\prod\limits_{j=1}^n \tfrac{\I + {U Q_j U^\dagger}}{2}.
	\label{eq:concrete-efficient-conjugation}
\end{talign}
While \cref{eq:concrete-efficient-conjugation} holds for any gate $U$, stabilizer
simulation can only exploit it if $UQ_jU^\dagger \in \pauligroup{n}$.
\emph{Clifford gates} such as $S$, $H$, $CNOT$, $\I$, $\X$, $\Y$,
and $\Z$ satisfy this for any $Q_j \in \pauligroup{n}$.

To also support the application of non-Clifford gates such as $T$ gates, we
follow \cite[\secref{VII.C}]{aaronson_improved_2004} and represent $\rho$ more
generally as
\begin{talign*}
\rho = \sum\limits_{i=1}^m
c_i
{\pauli_i}
\prod\limits_{j=1}^n \tfrac{\I + (-1)^{b_{ij}} {Q_j}}{2},
\end{talign*}
for $c_i \in \complex$, ${\pauli_i} \in \pauligroup{n}$,
$b_{ij} \in \bools$, and ${Q_j} \in \pauligroup{n}$.
Here, applying $U$ to $\rho$ amounts to replacing $P_i$ by
$U P_i U^\dagger$ and $Q_j$ by
$U Q_j U^\dagger$, which we can exploit if both
$U P_i U^\dagger$ and $U Q_j U^\dagger$
lie in $\pauligroup{n}$.

Otherwise, we decompose\footnote{This decomposition always exists and is unique,
as bare Pauli elements span (more than) $\unitary[2^n]$.} $U$ to the sum
$\sum_{p=1}^{K} d_p R_p$, where ${d_p \in \complex}$ and $R_p \in
\bare(\pauligroup{n})$ are bare Pauli elements, which have a prefactor of
$\i^0=1$.
Then,
\begin{talign}
	U \rho U^\dagger =& 
	\Big( 
		\sum\limits_{p=1}^K d_p R_p 
	\Big)
	\Big(
		\sum\limits_{i=1}^m
			{c_i}
			{\pauli_i}
			\prod\limits_{j=1}^n
			\tfrac{\I + (-1)^{b_{ij}} {Q_j}}{2}
	\Big)
	\Big(
		\sum\limits_{q=1}^K d_{q} R_{q} 
	\Big)^\dagger \\
	\overset{\text{\cite[\secref{VII.C}]{aaronson_improved_2004}}}{=}& 
	\sum\limits_{p=1}^K \sum\limits_{i=1}^m \sum\limits_{q=1}^K
		{c_{piq}}
		{P_{piq}}
			\prod\limits_{j=1}^n
			\tfrac{\I + (-1)^{b_{ijq}} {Q_j}}{2},
	\label{eq:concrete_app_U}
\end{talign}
for ${c_{piq}} = d_p c_i d_{q}^* \in \complex$, ${P_{piq}} = R_p P_i R_{q} \in
\pauligroup{n}$, and ${{b_{ijq}} = b_{ij} + \commut{Q_j}{R_{q}}} \in \bools$.
Here, $d_q^*$ denotes the complex conjugate of $d_q$, $+$ denotes addition
modulo $2$, and $\commut{Q_j}{R_{q}}$ is the commutator defined as~$0$ if $Q_j$
and $R_q$ commute and $1$ otherwise.
Note that $\commut{\cdot}{\cdot} \colon \pauligroup{n} \times \pauligroup{n} \to
\bools$ has the highest precedence.

Overall, the decomposition of a $k$-qubit non-Clifford gate results in at most
$K=4^k$ summands, thus blowing up the number of summands in our representation
by at most $4^k \cdot 4^k=16^k$.
In practice, the blow-up is typically smaller, e.g., decomposing a $T$ gate only
requires $2$ summands, while decomposing a $CCNOT$ gate requires $8$ summands.

\para{Measurement}
Measuring in bare Pauli basis ${P \in \bare(\pauligroup{n})}$ yields one of two
possible quantum states.
They can be computed by applying the two \emph{projections} $P_+ := \tfrac{\I +
P}{2}$ and $P_{-} = \tfrac{\I - P}{2}$, resulting in states $\rho_{+} = P_{+}
\rho P_{+}$ and $\rho_{-} = P_{-} \rho P_{-}$, respectively. For example,
collapsing the $i^\text{th}$ qubit to $\ket{0}$ or $\ket{1}$ corresponds to
measuring in Pauli basis $\applyat[i]{Z}$. The probability of outcome $\rho_{+}$
is $\ttrace{\rho_{+}}$, and analogously for $\rho_{-}$.
Note that we avoid renormalization for simplicity.
We discuss in \cref{sec:abstract-transformers} how measurements are
performed in stabilizer simulation~\cite[Sec.~VII.C]{aaronson_improved_2004}.

\para{Abstract Interpretation}
Abstract interpretation~\cite{cousotAbstract1977} is a framework for formalizing
approximate but sound calculation.
An \emph{abstraction} consists of ordered sets $(2^\cd, \subseteq)$ and $(\ad,
\leq)$, where $\cd$ and $\ad$ are called \emph{concrete set} and \emph{abstract
set} respectively together with a \emph{concretization function} $\gamma \colon
\ad \to 2^\cd$ which indicates which concrete elements $x=\gamma(\abe{x})
\subseteq \cd$ are represented by the abstract element $\abe{x}$.
Additionally, $\bottom \in \ad$ refers to $\emptyset = \gamma(\bottom) \subseteq
\cd$ and $\top \in \ad$ refers to $\cd = \gamma(\top)$.  

An abstract transformer $f^\sharp \colon \ad \to \ad$ of a function $f \colon
\cd \to \cd$ satisfies $\gamma \circ f^\sharp(\abe{x}) \supseteq f \circ
\gamma(\abe{x})$ for all $\abe{x} \in \ad$, where $f$ was lifted to operate on
subsets of $\cd$. 
This ensures that $f^\sharp$
(over-)approximates $f$, a property referred to as \emph{soundness} of
$f^\sharp$.
Abstract transformers can analogously be defined for functions $f \colon \cd^n
\to \cd$. Further, we introduce \emph{join} $\join \colon \ad \times \ad
\to \ad$, satisfying $\gamma(\abe{x}) \cup \gamma(\abe{y}) \subseteq \gamma(\abe{x} \join \abe{y})$.
Throughout this work, we distinguish abstract objects $\abe{x} \in \ad$ and
concrete objects $x \in \cd$ by stylizing them in bold or non-bold respectively.

As an example, a common abstraction is the interval abstraction with
$\cd=\mathbb{R}$. 
The abstract set is the set of intervals
$$
\ad = \{(l,u) \mid l, u \in \mathbb{R} \cup \{\pm\infty\} \},
$$
where $\abe{x} = (l,u)$ is a tuple.
The concretization function $\gamma \colon \ad \to \cd$ maps these tuples to
sets: $$\gamma(\abe{x}) = [l, u] = \{y \in \mathbb{R} \mid l \leq y \leq u\}.$$
Further, $\top = (-\infty,\infty)$ and $\bottom = (l,u)$
for $l>u$. 
Common abstract transformers for the interval abstraction are shown in \cref{tab:interval-transformers}.

\begin{table}
	\caption{Transformers for the interval abstraction.}
    \label{tab:interval-transformers}
	\small
    \centering
	\begin{tabular}{ @{} c l l @{} }
		\toprule
		Function & Abstract Transformer & Efficient Closed form \\
		\hline
		$+$ & $[l_1, u_1] +^\sharp [l_2, u_2] = [l', u']$ & $l' = l_1 + l_2$ and $u' = u_1 + u_2$ 
		\\
		$\cdot$ & $[l_1, u_1] \cdot^\sharp [l_2, u_2] = [l', u']$ & $l' = \min(l_1 l_2, l_1 u_2, u_1 l_2, u_1 u_2)$, $u'$ defined analogously with $\max$
		\\
		$\exp$ & $\exp^\sharp([l, u]) = [l', u']$ & $l' = \exp(l)$ and $u' = \exp(u)$ \\
		$\cos$ & $\cos^\sharp([l, u]) = [l', u']$ & exists, several case distinctions necessary \\
		$\cup$ & $[l_1, u_1] \join [l_2, u_2] = [l', u']$ & $l' = \min(l_1, l_2)$ and $u' = \max(u_1, u_2)$ \\
		\bottomrule
	\end{tabular}
\end{table}

The transformers in \cref{tab:interval-transformers} are \emph{precise}, meaning that for $f \colon \mathbb{R} \to \mathbb{R}$, we have that $f^\sharp((l, u)) = (\min_{l \leq v \leq u} f(v), \max_{l \leq v \leq u} f(v))$ and analogously for $f \colon \mathbb{R}^n \to \mathbb{R}$.
An abstract transformer for a composition of functions $f \circ g$ is the
composition of the abstract transformers. Although this is sound, it is not
necessarily precise: let $g \colon \mathbb{R} \to \mathbb{R}^2$ with $g(x) =
\begin{psmallmatrix} x \\ x \end{psmallmatrix}$ and $f \colon \mathbb{R}^2 \to
\mathbb{R}$ with $f(x,y)=x \cdot y$, then $f \circ g (x) = x^2$, but $f^\sharp
\circ g^\sharp((-2, 2)) = (-4, 4)$ whereas a precise transformer would map $(-2,
2)$ to $(0, 4)$. 

\para{Notational Convention}
In slight abuse of notation, throughout this work we may write the
concretization of abstract elements instead of the abstract element itself.
For example, for \mbox{$(0,1) \in \areals$}, we write $[0,1]$ defined as $\{v
\in \reals \mid 0 \leq v \leq 1 \}$ to indicate that it represents an interval. 
Where clear from context, we omit $\sharp$ and write $f$ for $f^\sharp$. For
example, we write $[l_1, u_1] + [l_2, u_2]$ for $[l_1, u_1] +^\sharp [l_2, u_2]$.
\section{Overview} \label{sec:overview}

In this section, we showcase \tool by applying it to the example circuit in
\cref{fig:overview}.
Overall, \tool proceeds analogously to
\cite[\secref{VII-C}]{aaronson_improved_2004}, but operates on abstract summands
representing many concrete summands.

\begin{figure*}
	\includegraphics[width=\linewidth]{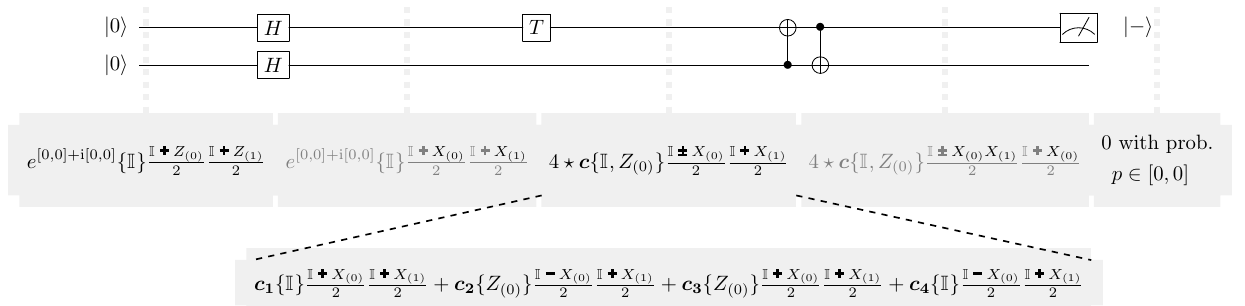}
	\caption{Overview of \tool, where we define $\abe{c}$ and
	$\abe{c_1}$--$\abe{c_4}$ in \cref{sec:overview}.}
\label{fig:overview}
\end{figure*}

\para{Example Circuit}
We first discuss the circuit in \cref{fig:overview}. Both qubits are
initialized to $\ket{0}$. The circuit then applies a succession of gates. The
abstract representation of the state after the application of each gate is shown
in the gray boxes below the circuit.
On the final state, the circuit collapses the upper qubit to $\ket{-}$ by
applying the projection $M_- = \tfrac{\I - \applyat[0]{X}}{2}$.
Precise circuit simulation shows that the probability of obtaining $\ket{-}$ is
$0$, in this case.
In the following, we demonstrate how \tool computes an over-approximation of
this probability. 

\para{Initial State}
The density matrix for the initial state $\ket{0} \otimes \ket{0}$ can be
represented as (see~\cite{aaronson_improved_2004}):
\begin{equation*}
	\rho_A 
	=  
	1 \scalmult \I  
	\tfrac{\I + (-1)^0 \applyat[0]{Z}}{2}
	\tfrac{\I + (-1)^0\applyat[1]{Z}}{2}.
\end{equation*}
To translate this to an abstract density matrix, we simply replace some elements
by abstract representations. This gives the following initial abstract state:
\begin{align} \label{eq:abstract-initial-full}
	\abe{\rho_A}
	= 
	e^{[0, 0] + [0, 0]\i} \scalmult 
	\inset{\I}
	\tfrac{\I + (-1)^{\inset{0}} \applyat[0]{Z}}{2}
	\tfrac{\I + (-1)^{\inset{0}} \applyat[1]{Z}}{2}.
\end{align}
Here we abstract booleans as sets, for instance $\{0\}$.
For conciseness, in \cref{fig:overview} we write $x \Plus y$, $x \Minus y$, and
$x \PlusMinus y$ for $x + (-1)^{\{0\}} y$, $x + (-1)^{\{1\}} y$, and $x +
(-1)^{\{0,1\}} y$. 
Further, we represent abstract complex numbers in polar form with logarithmic
length, using real intervals: $1$ is represented as $e^{[0, 0] + [0, 0]\i}$,
while we can over-approximate the set of complex numbers $\{1, \i\}$ as
$e^{[0, 0] + [0, \tfrac{\pi}{2}]\i}$. 
Finally, we abstract Pauli elements as sets, such as $\{\I\}$ in
\cref{fig:overview} and \cref{eq:abstract-initial-full}. In
\cref{sec:abstractElements}, we will clarify how we store these sets
efficiently, for example representing $\inset{\I}$ as $\i^\inset{0} \cdot \{\I\}
\otimes \{\I\}$ and $\inset{\i \cdot \I,\i \cdot \applyat[0]{\Z}}$ as
$\i^\inset{1} \cdot \{\I,\Z\} \otimes \{\I\}$.

We now explain how each operation in the circuit modifies this abstract state.

\para{Clifford Gate Application}
First, the circuit applies one Hadamard gate $H$ to each qubit. This corresponds
to the unitary operator $\applyat[0]{H}\applyat[1]{H}$, yielding updated
abstract density matrix $\abe{\rho_B} =
(\applyat[0]{H}\applyat[1]{H})\abe{\rho_A}
(\applyat[0]{H}\applyat[1]{H})^\dagger.$ Just as for concrete density matrices
(see \cref{sec:background}), this amounts to replacing
\begin{minipage}{\linewidth}
	\centering
	\def\arraystretch{1.4}
	\setlength{\tabcolsep}{2pt}
	\begin{tabular}{ccl}
		$\inset{\I}$ & by &
		$(\applyat[0]{H}\applyat[1]{H}) \inset{\I}
		(\applyat[0]{H}\applyat[1]{H})^\dagger = \inset{\I}$,
		\\
		$\applyat[0]{Z}$ & by &
		$(\applyat[0]{H}\applyat[1]{H}) \applyat[0]{Z}
		(\applyat[0]{H}\applyat[1]{H})^\dagger = \applyat[0]{X}$, and
		\\
		$\applyat[1]{Z}$ & by &
		$(\applyat[0]{H}\applyat[1]{H}) \applyat[1]{Z} (\applyat[0]{H}\applyat[1]{H})^\dagger = \applyat[1]{X}$.
	\end{tabular}
\end{minipage}
We hence get
$
	\abe{\rho_B}
	= 
	e^{[0,0] + [0,0]\i} \scalmult 
	\inset{\I}
	\tfrac{\I + (-1)^{\inset{0}} \applyat[0]{X}}{2}
	\tfrac{\I + (-1)^{\inset{0}} \applyat[1]{X}}{2}.
$

\para{Non Clifford Gate Application}
Next, the circuit applies gate $T$ on the upper qubit.
To this end, we again follow the simulation described in \cref{sec:background}.
We first decompose $T$ into Pauli elements:
$
	\applyat[0]{T} = d_1\I + d_2\applyat[0]{Z},
$
where $d_1 \approx e^{-0.1 + 0.4 \i}$ and ${d_2 \approx e^{-1.0 - 1.2 \i}}$.
Replacing $T$ with its decomposition, we can then write $\abe{\rho_T} =
T\abe{\rho_B}T^{\dagger}$, using \cref{eq:concrete_app_U}, as:
\begin{align*}
	\abe{\rho_T} 
	=
	&\left(d_1\I + d_2\applyat[0]{Z}\right) 
	\left(
		e^{[0,0] + [0,0]\i} \scalmult 
	\inset{\I}
	\tfrac{\I + (-1)^{\inset{0}} \applyat[0]{X}}{2}
	\tfrac{\I + (-1)^{\inset{0}} \applyat[1]{X}}{2}
	\right) \\
	&\quad\left(d_1 \I + d_2 \applyat[0]{Z}\right)^\dagger.
\end{align*}
Analogously to \cref{sec:background}, we can rewrite this to:
\begin{align*}
	&\abe{c_1} \scalmult  \inset{\I}  
		\tfrac{\I + (-1)^{\inset{0}}\applyat[0]{X}}{2}
		\tfrac{\I + (-1)^{\inset{0}}\applyat[1]{X}}{2} \\
	+& \abe{c_2} \scalmult  \inset{\applyat[0]{Z}}  
		\tfrac{\I + (-1)^{\inset{1}}\applyat[0]{X}}{2}
		\tfrac{\I + (-1)^{\inset{0}}\applyat[1]{X}}{2} \\
	+& \abe{c_3} \scalmult  \inset{\applyat[0]{Z}}    
	\tfrac{\I + (-1)^{\inset{0}}\applyat[0]{X}}{2}
	\tfrac{\I + (-1)^{\inset{0}}\applyat[1]{X}}{2} \\
	+& \abe{c_4} \scalmult  \inset{\I}   
	\tfrac{\I + (-1)^{\inset{1}}\applyat[0]{X}}{2}
	\tfrac{\I + (-1)^{\inset{0}}\applyat[1]{X}}{2},
\end{align*}
where
\begin{align*}
	&\abe{c_1} = d_1 e^{[0,0] + [0,0]\i} d_1^* \approx e^{[-0.2,-0.2] + [0,0]\i},\\
	&\abe{c_2} = d_1 e^{[0,0] + [0,0]\i} d_2^* \approx e^{[-1.1,-1.1] + [1.6,1.6]\i}, \\
	&\abe{c_3} = d_2 e^{[0,0] + [0,0]\i} d_1^* \approx e^{[-1.1,-1.1] + [-1.6,-1.6]\i}, \\
	&\abe{c_4} = d_2 e^{[0,0] + [0,0]\i} d_2^* \approx e^{[-2.0,-2.0] + [0,0]\i}.
\end{align*}	

\para{Merging Summands}
Unfortunately, simply applying $T$ gates as shown above may thus increase the number of summands in the abstract density matrix by a factor of $4$. 
To counteract this, our key idea is to merge summands, by allowing a single
abstract summand to represent multiple concrete ones, resulting in reduced
computation overhead at the cost of lost precision. Our abstract representation
allows for a straightforward merge: we take the union of sets and join
intervals. Specifically, for complex numbers, we join the intervals in their
representation, obtaining:
\begin{align*}
	\abe{c} := \abe{c_1} \join \abe{c_2} \join \abe{c_3} \join \abe{c_4} 
	&= e^{[-2.0, -0.2] + [-1.6, 1.6]\i}.
\end{align*}
Finally, we introduce the symbol $\star$ to denote how many concrete summands an
abstract summand represents. Altogether, merging the summands in $\abe{\rho_T}$
yields:
\begin{align*}
	\abe{\rho_C} &= 4 \star e^{[-2.0, -0.2] + [-1.6, 1.6]\i} \scalmult \inset{\I, \applyat[0]{Z}}  
	\tfrac{\I + (-1)^{\inset{0, 1}}\applyat[0]{X}}{2}
	\tfrac{\I + (-1)^{\inset{0}}\applyat[1]{X}}{2}.
\end{align*}
Note that for an abstract element $\abe{x}$, $r \star \abe{x}$ is not equivalent
to $r \cdot \abe{x}$. For example, $2 \star \{0,1\} = \{0,1\} + \{0,1\} =
\{0,1,2\}$, while $2 \cdot \{0,1\} = \{0,2\}$.~\footnote{We implicitly lift
concrete elements to abstract elements: $2 \cdot \{0,1\}=\{2\} \cdot
\{0,1\}=\{0,2\}$.}

\para{Measurement}
After the $T$ gate, the circuit applies two additional $CNOT$ gates, resulting in the updated density matrix:
\begin{equation*}
	\abe{\rho_D} 
	= 
	4 \star 
		e^{[-2.0, -0.2] + [-1.6, 1.6]\i} \scalmult  \inset{\I, \applyat[0]{Z}}  
	% \\
		\tfrac{\I + (-1)^{\inset{0, 1}}\applyat[0]{X}\applyat[1]{X}}{2}
		\tfrac{\I + (-1)^{\inset{0}}\applyat[0]{X}}{2}.
\end{equation*}
Finally, the circuit applies the projection $M_{-} = \tfrac{\I
-\applyat[0]{X}}{2}$.
To update the density matrix accordingly, we closely follow
\cite{aaronson_improved_2004}, which showed that measurement can be reduced to
simple state updates through a case distinction on $M_{-}$ and the state $\rho$.
If (i) the measurement Pauli (here~$-\applyat[0]{X}$) commutes with the product
Paulis (here $(-1)^{\inset{0, 1}}\applyat[0]{X}\applyat[1]{X}$ and
$(-1)^{\inset{0}}\applyat[1]{X}$) and (ii) the measurement Pauli cannot be
written as a product of the product Paulis, the density matrix after measurement
is $0$. We will explain in \cref{sec:measurement} how our abstract domain allows both of these checks to
be performed efficiently.

Here, both conditions are satisfied, and we hence get the final state
$\abe{\rho_{M1}} = 0$. We can then compute the probability of such an outcome by
$p = \ttrace{\abe{\rho_{M1}}} = 0$.
Thus, our abstract representation was able to provide a fully precise result.

\para{Imprecise Measurement}
Suppose now that instead of the measurement in \cref{fig:overview}, we had
collapsed the lower qubit to $\ket{0}$ by applying projection $M_0 = \tfrac{\I
+ \applyat[1]{Z}}{2}$.

To derive the resulting state, we again follow \cite{aaronson_improved_2004}
closely. We note that the measurement Pauli $+\applyat[1]{Z}$ (i) anticommutes
with the first product Pauli $(-1)^{\inset{0, 1}}\applyat[0]{X}\applyat[1]{X}$
and commutes with the second one $(-1)^{\inset{0}}\applyat[0]{X}$ and (ii)
commutes with the initial Paulis $\{\I, \applyat[0]{Z}\}$. In this case, we get
that the density matrix is unchanged, thus $\abe{\rho_{M2}} = \abe{\rho_D}$. To
compute the trace of this matrix, we follow the procedure outlined in
\cref{sec:trace}. We omit intermediate steps here and get:~\footnote{We used the
precise interval bounds for $\abe{c}$ here, not the rounded values provided
earlier.}
$$p = \ttrace{\abe{\rho_{M2}}} = 4 \scalmult \Re(\abe{c}) \approx [0, 1.7].$$
Thus, our abstraction here is highly imprecise and does not yield any
information on the measurement result (we already knew that the probability must
lie in $[0,1]$).

\section{Abstract Domains} \label{sec:abstractElements}
In the following, we formalize all abstract domains
(\cref{tab:abstract-domains}) underlying our abstract representation of density
matrices~$\abe{\rho}$ along with key abstract transformers operating on them
(\cref{tab:abstractops}).
We note that all abstract transformers introduced here naturally also support
(partially) concrete arguments.

\para{Example Elements}
\cref{tab:abstract-domains} provides an example element $\abe{x}$ of each
abstract domain, along with an example of its concretization $\gamma(\abe{x})$,
where ${\gamma \colon \ad \to 2^\cd}$.
While \cref{tab:abstract-domains} correctly
distinguishes abstract elements from their concretization, in the following,
when describing operators we write concretizations instead of abstract elements
(as announced in \cref{sec:background}).

\para{Booleans and $\mathbb{Z}_4$}
Abstract booleans $\abe{b} \in \abools = 2^\bools$ are subsets of $\bools$, as
exemplified in \cref{tab:abstract-domains}.
The addition of two abstract booleans naturally lifts boolean addition to sets
and is clearly sound:
\begin{align} \label{eq:addition-set}
	\abe{b} + \abe{c} = \{b + c \mid b \in \abe{b}, c \in \abe{c}\}.
\end{align}
We define multiplication of abstract booleans analogously.
Further, we define the join of two abstract booleans as their set union.

Analogously to booleans, our abstract domain $\azz$ consists of subsets of
$\mathbb{Z}_4$, where addition, subtraction, multiplication, and joins works analogously to abstract booleans.
Further, we can straight-forwardly embed abstract booleans into $\azz$ by mapping $0$ to $0$ and $1$ to $1$. %could have been 1 maps to 2

\begin{table}
	\caption{Example elements on abstract domains.}
	\label{tab:abstract-domains}
	\small
	\centering
	\begin{tabular}{@{}lll@{}l@{}}
		\toprule
		\textbf{Dom.} & \textbf{Example element} & \textbf{Concretization} \\ \midrule
		$\abools$ &
		$\{0,1\}$ &
		$\{0,1\}$ \\
		%%%%%%%%%%%%%%%%%%%%%%%%%%%%%%%%%%
		$\azz$ &
		$\{0,3\}$ &
		$\{0,3\}$ \\
		%%%%%%%%%%%%%%%%%%%%%%%%%%%%%%%%%%
		$\areals$ &
		$({\color{my-full-blue}0},{\color{my-full-orange}1})$ &
		$[{\color{my-full-blue}0},{\color{my-full-orange}1}]$ &
		$=\{r \mid {\color{my-full-blue}0} \leq r \leq {\color{my-full-orange}1} \}$ \\
		%%%%%%%%%%%%%%%%%%%%%%%%%%%%%%%%%%
		$\acomplex$
		& $({\color{my-full-blue}0},{\color{my-full-orange}1},{\color{my-full-green}\pi},{\color{my-full-red}2\pi})$
		& $e^{[{\color{my-full-blue}0},{\color{my-full-orange}1}]+[{\color{my-full-green}\pi},{\color{my-full-red}2\pi}]\i}$ &
		$= \{e^{r+\varphi\i} \mid {\color{my-full-blue}0} \leq r \leq {\color{my-full-orange}1}, {\color{my-full-green}\pi} \leq \varphi \leq {\color{my-full-red}2\pi} \}$ \\
		%%%%%%%%%%%%%%%%%%%%%%%%%%%%%%%%%%
		$\apaulis[2]$ &
		$(
			{\color{my-full-olive}\{0,3\}},
			{\color{my-full-purple}\{\Z,\Y\}},
			{\color{my-full-brown}\{\X\}}
		)$ &
		$
			\i^{{\color{my-full-olive}\{0,3\}}} \cdot
			{\color{my-full-purple}\{\Z,\Y\}} \otimes
			{\color{my-full-brown}{\{\X\}}}
			\,
		$ &
		$	
			= \left\{
				\i^b \cdot \pauli^{(1)} \otimes \pauli^{(2)} \;\middle|\;
				\begin{array}{l}
					b \in {\color{my-full-olive}\{0,3\}}, \\
					\pauli^{(1)} \in {\color{my-full-purple}\{\Z,\Y\}},
					\pauli^{(2)} \in {\color{my-full-brown}{\{\X\}}}
				\end{array}
			\right\}
		$
		\\ \bottomrule
	\end{tabular}
\end{table}

\begin{table}
    \caption{Summary of abstract transformers.}
    \label{tab:abstractops}
	\small
    \centering
    \begin{tabular}{@{}lll@{}}
		\toprule
		\textbf{Transformers} & \textbf{Domains} &
		\textbf{Definition}\\ \midrule
		$\abe{b} + \abe{c} \in \abools, \abe{b} \cdot \abe{c} \in \abools$ & $\abe{b}, \abe{c} \in \abools$ & 
		Lifting to sets, \cref{eq:addition-set} \\
		$\abe{b} \join \abe{c} \in \abools$ & $\abe{b}, \abe{c} \in \abools$ & $\abe{b} \cup \abe{c}$ \\
		$\abe{b} + \abe{c} \in \azz, \abe{b} - \abe{c} \in \azz, \abe{b} \cdot \abe{c} \in \azz$ & $\abe{b}, \abe{c} \in \azz$ & 
		Lifting to sets \\
		$\abe{b} \join \abe{c} \in \azz$ & $\abe{b}, \abe{c} \in \azz$ & 
		$\abe{b} \cup \abe{c}$ \\
		$\abe{b} \in \azz$ & $\abe{b} \in \abools$ & Embedding \\
		$\abe{c}	\cdot \abe{d} \in \acomplex$ & $\abe{c}, \abe{d} \in \acomplex$ & \cref{eq:complexmult} \\
		$\abe{c} \sqcup \abe{d} \in \acomplex$ & $\abe{c}, \abe{d} \in \acomplex$ & \cref{eq:complexjoin} \\
		$\Re(\abe{c}) \in \areals$ & $\abe{c} \in \acomplex^n$ & 
		\cref{eq:complexsum} \\
		$\i^{\abe{b}} \in \acomplex$ & $\abe{b} \in \abools$ & \cref{eq:i-exp} \\
		$\abe{PQ} \in \apaulis$ & $\abe{P, Q} \in \apaulis$ & \cref{eq:pauliprod} \\
		$\pref(\abe{P}\abe{Q}) \in \azz$ & $\abe{P}, \abe{Q} \in \apaulis$ & \cref{eq:prefactor-of-product}
		\\
		$\applyat[i]{U}\abe{P}\applyat[i]{U^\dagger} \in \apaulis$ & $U \in \unitary[2^k], \abe{P} \in \apaulis$ & \cref{eq:pauliconj} \\
		$\commut{\abe{P}}{\abe{Q}} \in \abools$ & $\abe{P}, \abe{Q} \in \apaulis$ & \cref{eq:paulicomm} \\
		$\abe{P} \sqcup \abe{Q} \in \apaulis$ & $\abe{P, Q} \in \apaulis$ & 
		\cref{eq:paulijoin} \\
		$(-1)^{\abe{b}} \cdot \abe{P}$ & $\abe{b} \in \abools, \abe{P} \in \apaulis$ & \cref{eq:pauliminusone} \\
		\bottomrule
    \end{tabular}
\end{table}

\para{Real Numbers}
We abstract real numbers by intervals of the form $[\underline{a},\overline{a}]
\subseteq \mathbb{R}\cup \{\pm\infty\}$, and denote the set of such intervals by
$\areals$. Here, $\underline{a}$ and $\overline{a}$ indicate the lower and upper
bounds of the interval, respectively.
Interval addition, interval multiplication, and the cosine and exponential
transformer on intervals are defined in their standard way, see
\cref{sec:background}.

\para{Complex Numbers}
We parametrize complex numbers $c \in \complex$ in polar coordinates (with
magnitude in $\log$-space), as $c = e^{r + \varphi \i}$ for $r, \varphi \in
\reals$. For example, we parametrize $0$ as $e^{-\infty + 0\i}$.

Based on this parametrization, we abstract complex numbers using two real
intervals for $r$ and $\varphi$ respectively, as exemplified in
\cref{tab:abstract-domains}. Formally, we interpret $\abe{c} \in \acomplex$ as
the set of all possible outcomes when instantiating both intervals:
\begin{align*}
	\gamma(\abe{c}) = e^{
		[\underline{r}, \overline{r}] +
		[\underline{\varphi}, \overline{\varphi}] \i
	} =
	\left\{
		e^{r + \varphi \i}
		\;\middle|\;
		r \in [\underline{r}, \overline{r}],
		\varphi \in [\underline{\varphi}, \overline{\varphi}]
	\right\}.
\end{align*}

We can compute the multiplication and join of two abstract complex numbers
$
\abe{c} = e^{
	[\underline{r},\overline{r}] +
	[\underline{\varphi}, \overline{\varphi}] \i
}
$
and
$
\abe{c'} = e^{
	[\underline{r}',\overline{r}'] +
	[\underline{\varphi}', \overline{\varphi}'] \i
}
$
as
\begin{align}
	\abe{c} \cdot \abe{c'}
	&=
	e^{
		[\underline{r} + \underline{r}', \overline{r} + \overline{r}'] +
		[
			\underline{\varphi} + \underline{\varphi}',
			\overline{\varphi} + \overline{\varphi}'
		] \i
	} \text{ and } \label{eq:complexmult} \\
	%%%%%%%%%%%%%%%%%%%%%%%%%%%%%%%%%%%
	\abe{c} \join \abe{c'}
	&=
	e^{
		[
			\min(\underline{r}, \underline{r}'),
			\max(\overline{r}, \overline{r}')
		] + 
		[
			\min(\underline{\varphi}, \underline{\varphi}'),
			\max(\overline{\varphi}, \overline{\varphi}')
		] \i
	}. \label{eq:complexjoin}
\end{align}
Again, simple arithmetic shows that
\crefrange{eq:complexmult}{eq:complexjoin} are sound.
We note that to increase precision, we could map complex numbers to a canonical
representation before joining them, by exploiting
$e^{r+\phi\i}=e^{r+(\phi+2\pi)\i}$ to ensure that $\underline{\varphi}$ lies in
$[0, 2\pi]$.

We compute the real part of an abstract complex number
$\abe{c}=e^{[\underline{r},\overline{r}] +
[\underline{\varphi},\overline{\varphi}] \i}$ as
\begin{talign} \label{eq:complexsum}
	\Re(\abe{c})
	&= 
	\exp([\underline{r}, \overline{r}]) \cdot 
		\cos([\underline{\varphi}, \overline{\varphi}]),
\end{talign}
where we rely on interval transformers to evaluate the right-hand side. The
soundness of \cref{eq:complexsum} follows from the standard formula to extract
the real part from a complex number in polar coordinates.
We will later use \cref{eq:complexsum} to compute $\ttrace{\abe{\rho}}$.
To this end, we also need the transformer
\begin{align} \label{eq:i-exp}
	\i^{\abe{b}} = \bigsqcup_{b \in \abe{b}} \{ \i^b \} \in \acomplex.
\end{align}

\para{Pauli Elements}
Recall that a Pauli element $P \in \pauligroup{n}$ has the form $P = \i^v \cdot
P^{(0)} \otimes \cdots \otimes P^{(n-1)}$, for $v$ in $\mathbb{Z}_4$ and
${P^{(k)} \in \{\I, \X, \Y, \Z\}}$.
We therefore parametrize $P$ as a prefactor $v$ (in $\log_\i$ space) and $n$
bare Paulis $P^{(k)}$.

Accordingly, we parametrize abstract Pauli elements $\abe{\pauli} \in \apaulis$
as $\i^{\abe{v}} \cdot \abe{\pauli^{(0)}} \otimes \cdots \otimes
\abe{\pauli^{(n-1)}}$, where $\abe{v} \in \azz$ is a set of possible prefactors
and $\abe{\pauli^{(k)}} \subseteq \{\X,\Y,\Z,\I[2]\}$ are sets of possible Pauli
matrices.
Formally, we interpret $\abe{\pauli}$ as the set of all possible outcomes when
instantiating all sets:
\begin{align*}
	\gamma(\abe{\pauli})
	=
	\left\{
		\i^{v} \cdot \otimes_{i=0}^{n-1} \pauli^{(i)} \;\middle|\; v \in \abe{v}, \pauli^{(i)} \in \abe{\pauli^{(i)}}
	\right\}.
\end{align*}
We define the product of two abstract Pauli elements as:
\begin{equation}
	\abe{P}
	\abe{Q}
	=
	\i^{
		\pref(\abe{P}\abe{Q})
	}
	\otimes\limits_{i=0}^{n-1} \bare\Big(\abe{P^{(i)}}\abe{Q^{(i)}}\Big).
	\label{eq:pauliprod}
\end{equation}
To this end, we evaluate the prefactor induced by multiplying Paulis as
\begin{align}\label{eq:prefactor-of-product}
	\pref(\abe{P}\abe{Q}) = \pref(\abe{P}) + \pref(\abe{Q}) + \sum_{i=1}^n \pref(\abe{P^{(i)}} \abe{Q^{(i)}}),
\end{align}
where we can evaluate the summands in the right-hand side of
\cref{eq:prefactor-of-product} by precomputing them for all possible sets of
Pauli matrices $\abe{P^{(i)}}$ and $\abe{Q^{(i)}}$.
Then, we compute the sum using \cref{eq:addition-set}.
Analogously, we can evaluate $\bare\Big(\abe{P^{(i)}}\abe{Q^{(i)}}\Big)$ by
precomputation.
The soundness of \cref{eq:pauliprod} follows from applying the multiplication
component-wise, and then separating out prefactors from bare Paulis.

We also define the conjugation of an abstract Pauli element $\abe{P}$ with
$k$-qubit gate $U$ padded to $n$ qubits  as:
\begin{align}
	\applyat[i]{U} \abe{\pauli} \applyat[i]{U^\dagger}
	&=
	\applyat[i]{U}
	\Big(
		\i^{\abe{v}} \cdot
		\abe{\pauli^{(0:i)}} \otimes
		\abe{\pauli^{(i:i+k)}} \otimes
		\abe{\pauli^{(i+k:n)}}
	\Big)
	\applyat[i]{U^\dagger} \nonumber 
	\\
	&=
	\i^{
		\abe{v} +
		\pref(U\abe{\pauli^{(i:i+k)}}U^\dagger)
	} \cdot
	\abe{\pauli^{(0:i)}} \otimes
	\bare(U \abe{\pauli^{(i:i+k)}} U^\dagger) \otimes
	\abe{\pauli^{(i+k:n)}}, \label{eq:pauliconj}
\end{align}
where $\abe{\pauli^{(i:j)}}$ denotes $\abe{\pauli^{(i)}} \otimes \cdots \otimes
\abe{\pauli^{(j-1)}}$. Because $k$ is typically small, and all possible gates
$U$ are known in advance, we can efficiently precompute
$\pref(U\abe{\pauli^{(i:i+k)}}U^\dagger)$ and $\bare(U \abe{\pauli^{(i:i+k)}}
U^\dagger)$.
We note that this only works if the result of conjugation is indeed an
(abstract) Pauli element---if not, this operation throws an error\footnote{We
can recover from this error by decomposing $U$ as a sum of bare Pauli elements,
as mentioned in \cref{sec:background}, see also
\crefrange{eq:efficient-gate-application}{eq:inefficient-gate-application}.}.
The soundness from \cref{eq:pauliconj} follows from applying $U$ to qubits $i$
through $i+k$, and then separating out prefactors from bare Paulis.

We define the commutator $\commut{\abe{P}}{\abe{Q}}$ of two abstract Pauli
elements $\abe{P}$ and $\abe{Q}$ as
\begin{align}
	\commut{
		\Big(
			\i^{\abe{v}} \cdot
			\otimes\limits_{i=0}^{n-1}
			\abe{\pauli^{(i)}}
		\Big)
	}{
		\Big(
			\i^{\abe{w}} \cdot
			\otimes\limits_{i=0}^{n-1}
			\abe{Q^{(i)}}
		\Big)
	}
	=
	\sum_{i=1}^n \commut{\abe{\pauli^{(i)}}}{\abe{Q^{(i)}}}. \label{eq:paulicomm}
\end{align}
Here, we evaluate the sum using \cref{eq:addition-set}, and efficiently
evaluate $\commut{\abe{\pauli^{(i)}}}{\abe{Q^{(i)}}} \in \abools$ by precomputing:
\begin{align*}
	\commut{\abe{\pauli^{(i)}}}{\abe{Q^{(i)}}}
	=
	\left\{
		\commut{\pauli^{(i)}}{Q^{(i)}} \;\middle|\; \pauli^{(i)} \in \abe{\pauli^{(i)}}, Q^{(i)} \in \abe{Q^{(i)}}
	\right\}.
\end{align*}
The soundness of \cref{eq:paulicomm} can be derived from the corresponding
concrete equation, which can be verified using standard linear algebra.

We define the join of abstract Pauli elements as
\begin{equation}
		\Big(
			\i^{\abe{v}}
			\otimes\limits_{i=0}^{n-1}
			\abe{\pauli^{(i)}}
		\Big)
	\join
		\Big(
			\i^{\abe{w}}
			\otimes\limits_{i=0}^{n-1}
			\abe{Q^{(i)}}
		\Big)
	=
		\i^{\abe{v} \join \abe{w}}
		\otimes\limits_{i=0}^{n-1}
		\Big( \abe{P^{(i)}} \cup \abe{Q^{(i)}} \Big),
	\label{eq:paulijoin}
\end{equation}
where $\abe{P^{(i)}} \cup \abe{Q^{(i)}} \subseteq \{\I,\X,\Y,\Z\}$.
Clearly, this join is sound.

Finally, we define an abstract transformer for modifying the sign of an abstract
Pauli element $\abe{P}$ by:
\begin{align}
	(-1)^{\abe{b}} \cdot
	\Big(
		\i^{\abe{v}} \cdot
		\otimes_{i=0}^{n-1}
		\abe{\pauli^{(i)}}
	\Big)
	=
	\i^{\abe{v} + 2 \cdot \abe{b}} \cdot
	\otimes_{i=0}^{n-1}
	\abe{\pauli^{(i)}}
	\label{eq:pauliminusone}
\end{align}
The soundness of \cref{eq:pauliminusone} follows directly from $(-1)^v =
\i^{2v}$. 

\para{Abstract Density Matrices}
The concrete and abstract domains introduced previously allow us to represent an
abstract density matrix $\abe{\rho} \in \adensities$ as follows:
\begin{align}
	\abe{\rho} 
	= 
		{r} \star
		{\abe{c}}
		\cdot
		{\abe{P}}
		\cdot 
		\prod_{j=1}^n 
			\tfrac{\I + (-1)^{\abe{b}_{j}} {Q_j}}{2}.
	\label{eq:rho-abstract}
\end{align}
Here, $r \in \mathbb{N}$, $\abe{c} \in \acomplex$, $\abe{P} \in \apaulis$,
$\abe{b}_{j} \in \abools$, and $Q_{j} \in \pauligroup{n}$. Note that $Q_j$ are
concrete Pauli elements, while $\abe{P}$ is abstract.
Further, both $\abe{P}$ and $Q_{j}$ can have a prefactor, i.e., are not
necessarily bare Paulis.
Here, the integer counter $r$ records how many concrete summands were
abstracted. Specifically, $r \star \abe{x}$ is defined as $\sum_{i=1}^r
\abe{x}$.
Overall, we interpret $\abe{\rho}$ as:
\begin{equation}
	\gamma(\abe{\rho}) = \left\{ 
		\sum\limits_{i = 1}^{r}
		c_{i} P_{i}
		\prod\limits_{j=1}^n 
		\tfrac{\I + (-1)^{b_{ij}} {Q_j}}{2}
		\;\middle|\;
		c_{i} \in \gamma(\abe{c}), 
		P_{i} \in \gamma(\abe{P}), 
		b_{ij} \in \gamma(\abe{b}_{j})
	\right\},
\end{equation}
relying on the previously discussed interpretations of {$\complex$},
{$\pauligroup{n}$}, and {$\bools$}.

\section{Abstract Transformers} \label{sec:abstract-transformers}
We now formalize the abstract transformers used by \tool to simulate quantum
circuits.
The soundness of all transformers is straightforward, except for the trace
transformer (\cref{sec:trace}) which we discuss in \cref{app:soundness}.

\para{Initialization}
We start from initial state $\otimes_{i=1}^n \ket{0}$, which corresponds to
density matrix
\begin{equation*}
	\abe{\rho}
	\overset{}{=}
	\prod\limits_{j=1}^n 
		\tfrac{\I + \applyat[j]{Z}}{2}
	= 
	1 \star
	e^{[0,0] + \i[0,0]} \cdot \i^{\{0\}} \{\I\} 
		\prod\limits_{j=1}^n 
		\tfrac{\I + (-1)^{\{0\}} \applyat[j]{Z}}{2},
\end{equation*}
as established in \cite[Sec.~III]{aaronson_improved_2004}.
We note that we can prepare other starting states by applying appropriate gates
to the starting state $\otimes_{i=1}^n \ket{0}$.

\subsection{Gate Application}
Analogously to the concrete case discussed in \cref{sec:background}, applying a
unitary gate $U$ to $\abe{\rho}$ yields:
\begin{equation}
	U \abe{\rho} U^{\dagger}  
	= 
		r \star
		\abe{c} \abe{P'}
		\prod_{j=1}^n 
		\tfrac{\I + (-1)^{\abe{b}_{j}} Q'_j}{2}, \label{eq:efficient-gate-application}
\end{equation}
for $\abe{P'} = U \abe{P} U^\dagger$ and $Q'_j = U Q_j U^{\dagger}$.

If either $U \abe{P} U^\dagger \not\subseteq \pauligroup{n}$ or $U Q_j
U^\dagger \not\subseteq \pauligroup{n}$, \cref{eq:efficient-gate-application}
still holds, but we cannot represent the resulting matrices efficiently.
In this case, again analogously to \cref{sec:background}, we instead decompose
the offending gate as $U = \sum_{p} d_p R_p$, with $R_p \in \pauligroup{n}$ and obtain
\begin{equation}
	U \abe{\rho} U^\dagger = 
	\sum\limits_{pq}
		r \star \abe{c_{pq}}
		\abe{P_{pq}}
			\prod\limits_{j=1}^n
			\tfrac{\I + (-1)^{\abe{b_{jq}}} {Q_j}}{2},
	\label{eq:inefficient-gate-application}
\end{equation}
for $\abe{c}_{pq} = d_p \abe{c} d_q^*$, $\abe{P}_{pq}' =  R_p \abe{P}
R_{q}$, and $\abe{b_{jq}} = {\abe{b_{j}} + \commut{Q_j}{R_{q}}}$.

Overall, we can evaluate
\crefrange{eq:efficient-gate-application}{eq:inefficient-gate-application} by
relying on the abstract transformers from \cref{sec:abstractElements}.

\para{Compression}
To prevent an exponential blow-up of the number of summands and to adhere to the
abstract domain of $\abe{\rho}$ which does not include a sum, we compress all
summands to a single one. Two summands can be joined as follows:
\begin{align*}
	&\left(
		r_1 \star
		\abe{c}_1 \abe{P}_1 \prod_{j=1}^n 
		\tfrac{\I + (-1)^{\abe{b}_{1j}} Q_j}{2}
	\right)
	\sqcup
	\left(
		r_2 \star
		\abe{c}_2 \abe{P}_2 \prod_{j=1}^n 
		\tfrac{\I + (-1)^{\abe{b}_{2j}} Q_j}{2}
	\right)
	\\
	&=
		r \star
		\abe{c}
		\abe{P}
		\prod_{j=1}^n 
		\tfrac{\I + (-1)^{\abe{b}_{j}} Q_j}{2},
\end{align*}
where $r = r_1 + r_2$, $\abe{c} = \abe{c}_1 \sqcup \abe{c}_2$, $\abe{b}_{j} =
\abe{b}_{1j} \sqcup \abe{b}_{2j}$, and $\abe{P} = \abe{P}_1 \join \abe{P}_2$.
The key observation here is that the concrete $Q_j$ are independent of the
summand, and thus need not be joined.

We note that we could also only merge \emph{some} summands and leave the others
precise---investigating the effect of more flexible merging strategies could be
interesting future research.

\subsection{Measurement}\label{sec:measurement}
We now describe how to perform Pauli measurements, by extending the (concrete)
stabilizer simulation to abstract density matrices.
The correctness of the concrete simulation was previously established in
~\cite[Sec.~VII.C]{aaronson_improved_2004}, while the correctness of the abstraction is immediate. 

\para{Simulating Measurement}
Applying a  Pauli measurement in basis $R \in \bare(\pauligroup{n})$ has a
probabilistic outcome and transforms $\rho$ to $\rho_+ = \tfrac{\I +
R}{2}\rho\tfrac{\I+R}{2}$ with probability $\trace(\rho_+)$ or $\rho_- =
\tfrac{\I - R}{2}\rho\tfrac{\I-R}{2}$ with probability $\trace(\rho_-)$. We
describe how to compute $\rho_+$. Computing $\rho_-$ works analogously by using
$-R$ instead of $R$.

In the following, we will consider a concrete state $\rho$ as defined in \cref{sec:background} and an abstract state $\abe{\rho}$ as defined in \cref{eq:rho-abstract}:
\begin{align} 
	\rho
	=
	\sum_{i=1}^m c_i P_i
		\prod_{j=1}^n
		\tfrac{\I + (-1)^{b_{ij}} Q_j}{2}
	\quad\text{and}\quad
	\abe{\rho}
	=
	r \star \abe{c} \abe{P}
		\prod_{j=1}^n
		\tfrac{\I + (-1)^{\abe{b_{j}}} Q_j}{2}.
\end{align}

Concrete simulation of measurement distinguishes two cases: either (i) $R$ commutes with all $Q_j$ or (ii) $R$ anti-commutes with at least one $Q_j$. Note that as the $Q_j$ are concrete in an abstract state $\abe{\rho}$, those two cases translate directly to the abstract setting.
We now describe both cases for concrete and abstract simulation.

\para{Background: Concrete Case (i)}
In this case, we assume $R$ commutes with all $Q_j$. Focusing on a single
summand $\rho_i$ of $\rho$, measurement maps it
to (see~\cite{aaronson_improved_2004}):
\begin{equation} \label{eq:rewrite1concretecase1}
	\rho_{i,+} = c_i  \tfrac{\I + R}{2}
	P_i \tfrac{\I + R}{2} 
	\prod_{j=1}^n \tfrac{\I + (-1)^{b_{ij}} Q_j}{2}.
\end{equation}

Let us first introduce the notation $\{(-1)^{b_{ij}}Q_j\} \rightsquigarrow R$,
denoting that $R$ can be written as a product of selected Pauli elements from $\{(-1)^{b_{ij}}Q_j\}$.
Symmetrically, we write $\{(-1)^{b_{ij}}Q_j\} \not\rightsquigarrow R$ if $R$
cannot be written as such a product.
As shown in \cite{aaronson_improved_2004}, if $\{(-1)^{b_{ij}}Q_j\}
\rightsquigarrow R$ then $\tfrac{\I + R}{2}\prod_{j=1}^n \tfrac{\I +
(-1)^{b_{ij}} Q_j}{2}$ is equal to $\prod_{j=1}^n \tfrac{\I +
(-1)^{b_{ij}} Q_j}{2}$ and if $\{(-1)^{b_{ij}}Q_j\}
\not\rightsquigarrow R$ then $\tfrac{\I + R}{2}\prod_{j=1}^n \tfrac{\I +
(-1)^{b_{ij}} Q_j}{2}$ is null. 
Further, using that $R^2 = \I$, we get from \cref{eq:rewrite1concretecase1} that
if $P_i$ commutes with $R$, $\rho_{i, +}$ is equal to $\rho_i$, otherwise, $P_i$
anti-commutes with $R$ and $\rho_{i, +}$ is null.
Putting it all together, we finally get:
\begin{equation} \label{eq:measure-concrete-i}
	\rho_+ = \sum_{i=1}^m \rho_{i,+} = \sum_{i=1}^m
	\begin{cases}
		c_i P_i
		\prod\limits_{j=1}^n
		\tfrac{\I + (-1)^{b_{ij}} Q_j}{2} & \text{if } \{(-1)^{b_{ij}}Q_j\} \rightsquigarrow R \text{ and } \commut{R}{{P_i}} = 0,\\
		0 & \text{if } \{(-1)^{b_{ij}}Q_j\} \not\rightsquigarrow R \text{ or } \commut{R}{{P_i}} = 1.
	\end{cases}
\end{equation}

\para{Abstract Case (i)}
Let us first define $\rightsquigarrow^u$ and $\not\rightsquigarrow^u$ for a
concrete $R$, concrete $Q_j$ and abstract ${\abe{b_{j}}}$.
We say $\{(-1)^{\abe{b_{j}}}Q_j\} \rightsquigarrow^u R$ if for all $j$, for
all $b_{j} \in \gamma({\abe{b_{j}}})$, we have $\{(-1)^{{b_{j}}}Q_j\}
\rightsquigarrow R$.
Similarly, we say $\{(-1)^{\abe{b_{j}}}Q_j\} \not\rightsquigarrow^u R$ if for
all $j$, for all $b_{j} \in \gamma({\abe{b_{j}}})$, we have
$\{(-1)^{{b_{j}}}Q_j\} \not\rightsquigarrow R$.
Note that $\rightsquigarrow^u$ and $\not\rightsquigarrow^u$ are
under-approximations, and there can exist some $R$ and $\{(-1)^{\abe{b_{j}}}Q_j\}$
such that neither apply.
Using those two abstract relations, we get the abstract transformer for $\abe{\rho_+}$:
\begin{talign}\label{eq:measure-abstract-i}
	r \star \begin{cases}
		 \abe{c} \abe{P}
		\prod\limits_{j=1}^n
		\tfrac{\I + (-1)^{\abe{b_{j}}} Q_j}{2} & \text{if } \{(-1)^{\abe{b_{j}}}Q_j\} \rightsquigarrow^u R \text{ and } \commut{R}{\abe{P}} = \{0\},\\
		0 & \text{if } \{(-1)^{\abe{b_{j}}}Q_j\} \not\rightsquigarrow^u R \text{ or } \commut{R}{\abe{P}} = \{1\},\\
		\left(\abe{c}\sqcup \{0\}\right) \abe{P}
		\prod\limits_{j=1}^n
		\tfrac{\I + (-1)^{\abe{b_{j}}} Q_j}{2} & \text{otherwise.}
	\end{cases}
\end{talign}

We can evaluate \cref{eq:measure-abstract-i} by relying on the abstract
transformers from \cref{tab:abstractops} and by evaluating~$\rightsquigarrow^u$ as discussed shortly.

\para{Background: Concrete Case (ii)}
We now suppose $R$ anti-commutes with at least one $Q_j$. In this case, we can rewrite $\rho$ such that $R$ anti-commutes with $Q_1$, and commutes with all other $Q_j$.
Specifically, we can select any $Q_{j^*}$ which anti-commutes with $R$, swap
$b_{ij^*}$ and $Q_{j^*}$ with $b_{i1}$ and $Q_1$, and replace all other $Q_j$
anti-commuting with $R$ by $Q_1 Q_j$ (and analogously $b_{ij}$ by $b_{ij} +
b_{i1}$), which leaves $\rho$ invariant (see~\cite{aaronson_improved_2004}).
Assuming $\rho$ is the result after this rewrite, we have:
\begin{align}	\label{eq:measure-concrete-ii}
	\rho_+
	&=
	\sum_i \tfrac{1}{2} c_i
	P_i'
	\tfrac{\I + (-1)^{0} R}{2}
	\prod\limits_{j=2}^n \tfrac{\I + (-1)^{b_{ij}} Q_j}{2}, \\
	&\text{ where }  
	P_i' = \begin{cases}
		P_i & \text{if } \commut{R}{P_i} = 0, \\
		(-1)^{b_{i1}} P_i Q_1 & \text{if } \commut{R}{P_i} = 1.
	\end{cases}\nonumber
\end{align}
Overall, after rewriting $\rho$ as above, \cref{eq:measure-concrete-ii}
replaces $c_i$ by $\tfrac{1}{2}c_i$, $P_i$ by $P_i'$, $b_{i1}$ by $0$, and $Q_1$
by $R$.

\para{Abstract Case (ii)}
In the abstract case, we first apply the same rewrite as in the concrete case,
where we pick $j^*$ as the first $j$ for which $Q_j$ anti-commutes with
$R$.~\footnote{We could also consider other strategies than picking the first
possible $j$, for example picking a $j$ for which $\abe{b_j}$ is precise
whenever possible, to increase precision.} Then, directly abstracting
\cref{eq:measure-concrete-ii} yields:
\begin{align}\label{eq:measure-abstract-ii}
	\abe{\rho_+} &= r \star
	\tfrac{1}{2} \abe{c} \abe{P'} \tfrac{\I + (-1)^{\{0\}} R}{2}
	\prod\limits_{j=2}^n \tfrac{\I + (-1)^{\abe{b_{j}}} Q_j}{2}, \\
	\text{ where } \abe{P'}& = 
	\begin{cases}
		\abe{P} &\text{if } \commut{R}{\abe{P}} = \{0\}, \\
		(-1)^{\abe{b_{1}}} \abe{P} Q_1  & \text{if } \commut{R}{\abe{P}} = \{1\}, \\
		\abe{P} \sqcup (-1)^{\abe{b_{1}}} \abe{P} Q_1 & \text{otherwise.}
	\end{cases}\nonumber
\end{align}
Here, we replace $\abe{c}$ by $\tfrac{1}{2}\abe{c}$, $\abe{P}$ by $\abe{P'}$,
$\abe{b_{1}}$ by $\{0\}$, and $Q_1$ by $R$. When defining $\abe{P'}$, we follow
the two cases from \cref{eq:measure-concrete-ii} when our abstraction
is precise enough to indicate which case we should choose, or join the results
of both cases otherwise.
Again, we can evaluate \cref{eq:measure-abstract-ii} by relying on the
abstract transformers from \cref{tab:abstractops}.

\para{Joining Both Measurement Results}
For measurements occurring within a quantum circuit, stabilizer simulation
generally requires randomly selecting either $\rho_+$ or $\rho_-$ with
probability $\tr(\rho_+)$ and $\tr(\rho_-)$, respectively, and then continues
only with the selected state.
In contrast, \tool can join both measurement outcomes into a single abstract
state $\abe{\rho_+} \join \abe{\rho_-}$, as the $Q_j$ are the same in both. This
allows us to pursue both measurement outcomes simultaneously, as we demonstrate
in~\cref{sec:eval}.

\subsection{Efficiently computing $\rightsquigarrow$}
\label{sec:efficient-comp-squigg}
To simulate the result of a measurement, we introduced the new operator
$\{(-1)^{b_{j}}Q_j\} \rightsquigarrow R$, denoting that some Pauli $R$ can be
written as a product of $\{(-1)^{b_{j}}Q_j\}$. We now show how to compute
$\rightsquigarrow$ efficiently. 

\para{Background: Concrete case}
We first note that $\{(-1)^{b_{j}}Q_j\} \rightsquigarrow R$ holds if and only if there exist some $x \in \bools^n$ such that:
\begin{equation}\label{eq:non_bare}
R \overset{!}{=} \prod_{j=1}^n \left((-1)^{b_{j}}Q_j\right)^{x_j}.
\end{equation}
Further, this solution $x$ would satisfy:
\begin{equation}\label{eq:bare}
	\bare(R) \overset{!}{=} \bare\left(\prod_{j=1}^n \left((-1)^{b_{j}}Q_j\right)^{x_j}\right)
\end{equation}
\cref{eq:bare} has a solution if and only if $R$ commutes with all the $Q_j$, in
which case this solution $x$ is unique (see~\cite{aaronson_improved_2004}).
Hence, to check if $\{(-1)^{b_{j}}Q_j\} \rightsquigarrow R$, we can first verify
whether $\commut{R}{Q_j} = 0$ for all $j$, and if so, check if the unique $x$
satisfying \cref{eq:bare} also satisfies \cref{eq:non_bare}.

\para{Background: Finding $x$ for \cref{eq:bare}}
To compute this solution $x$, the stabilizer simulation relies critically on an
isomorphism $\encpauli$ between Pauli matrices $\{\I, \X, \Y, \Z\}$ and
$\bools^{2}$.

Specifically, $\encpauli$ maps $I$ to $\Ie$, $X$ to $\Xe$, $Y$ to $\Ye$, and $Z$
to $\Ze$. Further, $\encpauli$ extends naturally to bare Pauli elements $R \in
\bare{\left(\pauligroup{n}\right)}$ and tuples $Q=(Q_1,\dots,Q_n) \in
\bare{\left(\pauligroup{n}\right)}^n$ by:
\begin{align*}
	\encpauli(R) = \begin{psmallmatrix}
		\encpauli(R^{(0)}) \\
		\vdots \\
		\encpauli(R^{(n-1)})
	\end{psmallmatrix} \text{ and }
	\encpauli(Q) = \begin{psmallmatrix}
		\encpauli(Q_1^{(0)}) & \cdots & \encpauli(Q_n^{(0)}) \\
		\vdots & \ddots & \vdots \\
		\encpauli(Q_1^{(n-1)}) & \cdots & \encpauli(Q_n^{(n-1)}) \\
	\end{psmallmatrix},
\end{align*}
where $\encpauli(R) \in \bools^{2n \times 1}$ and $\encpauli(Q) \in \bools^{2n
\times n}$.
We can naturally extend $\encpauli$ to $\pauligroup{n}$, by defining
$\encpauli(R) = \encpauli(\bare(R))$.

This isomorphism $\encpauli$ is designed so that the product of bare Pauli elements ignoring prefactors
corresponds to a component-wise addition of encodings:
\begin{equation} \label{eq:embedding}
	\encpauli(\pauli_1 \pauli_2) = \encpauli(\pauli_1) + \encpauli(\pauli_2).
\end{equation}

Using \cref{eq:embedding}, we can obtain
solution candidates $x$ for \cref{eq:bare} by solving a system of linear equations using
Gaussian elimination modulo~$2$:
\begin{talign}
	\encpauli \left( R
	\right) \overset{!}{=} \encpauli \left( \prod\limits_{j=1}^n Q_j^{x_j} \right) =
	\sum\limits_{j=1}^n \encpauli(Q_j)x_j = \encpauli(Q)x. \label{eq:reduce-to-linear}
\end{talign}
Because in our case, $\encpauli(Q)$ is over-determined and has full rank,
\cref{eq:reduce-to-linear} either has no solution, or a unique solution $x$.

\para{Background: Checking prefactors}
Once we have found the unique $x$ (if it exists) satisfying \cref{eq:bare} as described above, we need to check if it also satisfies \cref{eq:non_bare}. It is enough to check if the prefactors match:
\begin{equation*}
	\pref\left(R\right) \overset{!}{=} \pref\left(\prod_j (-1)^{b_j
x_j}Q_j^{x_j}\right),
\end{equation*}
or equivalently:
\begin{equation*}
	\pref\left(R\right) - \pref\left(\prod_j Q_j^{x_j}\right) - 2 \sum_j b_j x_j \overset{!}{=} 0,
\end{equation*}
where the subtraction and sum operations are over $\mathbb{Z}_4$.

Putting it all together, we can define $\grppref \colon \pauligroup{n} \times \pauligroup{n}^n \times \bools^n \to \mathbb{Z}_4 \cup \{\lightning\}$ with
\begin{talign}\label{eq:concrete_grppref}
	\grppref(R,Q,b) =
	\begin{cases}
		\lightning & \text{if } \exists j, \commut{R}{Q_j} = 1, \\
		\pref(R) - \pref\left(\prod\limits_{j=1}^n Q_j^{x_j}\right) - 2\sum\limits_{j=1}^n x_j b_{j} & \text{otherwise},
	\end{cases}
\end{talign}
% }
where $x$ is the unique value such that $\encpauli(R) = \encpauli(Q)x$ and
$\lightning$ indicates there is no such $x$.
We then have that $\{(-1)^{b_{j}}Q_j\} \rightsquigarrow R$ if and only if
$\grppref(R,Q,b) = 0$, or equivalently, $\{(-1)^{b_{j}}Q_j\}
\not\rightsquigarrow R$ if and only if $\grppref(R,Q,b) \neq 0$.

\para{$\grppref$ for abstract $\abe{b_{j}}$} For abstract values $\abe{b_{j}}$,
we define $\grppref \colon \pauligroup{n} \times \pauligroup{n}^n \times
\abools^n \to 2^{\mathbb{Z}_4 \cup \{\lightning\}}$ as follows:
\begin{equation}\label{eq:abstract_b_grppref}
	\grppref(R,Q,\abe{b}) =
	\begin{cases}
		\{\lightning\} & \text{if } \exists j, \commut{R}{Q_j} = 1, \\
		\pref(R) - \pref\left(\prod\limits_{j=1}^n Q_j^{x_j}\right) - 2\sum\limits_{j=1}^n x_j \abe{b_{j}} & \text{otherwise.}
	\end{cases}
\end{equation}

Following the same reasoning as above, we have $\{(-1)^{\abe{b_{j}}}Q_j\}
\rightsquigarrow^u R$ if and only if $\grppref(R,Q,\abe{b}) = \{0\}$ and
$\{(-1)^{\abe{b_{j}}}Q_j\} \not\rightsquigarrow^u R$ if and only if
\mbox{$\grppref(R,Q,\abe{b}) \cap \{0\} = \emptyset$}.

\para{$\grppref$ for abstract $\abe{b_{j}}$ and $\abe{R}$}
To compute the trace of a state (see \cref{sec:trace}), we further extend
\cref{eq:concrete_grppref} to abstract $\abe{b_{j}}$ and abstract $\abe{R}$,
and define $\grppref \colon \apaulis[n] \times \pauligroup{n}^n \times \abools^n
\to 2^{\mathbb{Z}_4 \cup \{\lightning\}}$ as:
\begin{align}
	\grppref(\abe{R},Q,\abe{b}) &=
	\begin{cases}
		\{\lightning\} & \text{if } \exists j. \commut{\abe{R}}{Q_j} = \{1\}, \\
		\pref(\abe{R}) - \pref\left(\prod\limits_{j=1}^n Q_j^{\abe{x_j}}\right) - 2\sum\limits_{j=1}^n \abe{x_j} \abe{b_{j}} & \text{if } \forall j. \commut{\abe{R}}{Q_j} = \{0\}, \\
		\pref(\abe{R}) - \pref\left(\prod\limits_{j=1}^n Q_j^{\abe{x_j}}\right) - 2\sum\limits_{j=1}^n \abe{x_j} \abe{b_{j}} \cup \{\lightning\} & \text{otherwise,}
	\end{cases} \label{eq:f-abstract}
	\\ 
	\text{ for } \encpauli(\abe{R}) &= \encpauli(Q)\abe{x}. \label{eq:abstract-solve}
\end{align}
Here, evaluating \cref{eq:f-abstract} requires evaluating
$Q_j^{\abe{b}}$ for an abstract boolean $\abe{b}$, which we define naturally
as
\begin{equation*}
	Q_j ^{\abe{b}} 
	:= 
	\begin{cases}
		\{Q_j\}  & \text{if } \abe{b} = \{1\},
		\\
		\{\I\} & \text{if } \abe{b} = \{0\},
		\\
		\{Q_j, \I\} & \text{if } \abe{b} = \{0,1\}.
	\end{cases}
\end{equation*}
Further, \cref{eq:abstract-solve} requires over-approximating all $x$
which satisfy the linear equation ${\encpauli(\abe{R}) = \encpauli(Q) x}$. Here,
we naturally extend $\encpauli$ to abstract Paulis by joining their images. For
instance, we have that $\encpauli(\{X, Y\}) = \left\{\Xe\right\} \sqcup
\left\{\Ye\right\} = \big(\begin{smallmatrix} \{1\} \\ \{0, 1\}
\end{smallmatrix}\big)$.
We then view ${\encpauli(\abe{R}) = \encpauli(Q) x}$ as a system of linear
equations $\abe{b}=Ax$, where the left-hand side consists of abstract booleans
$\abe{b} \in \abools^{2n}$. We then drop all equations in this equation system
where the left-hand side is $\{0,1\}$, as they do not constrain the solution
space.
This updated system is fully concrete, hence we can solve it using Gaussian
elimination.
We get either no solution, or a solution space $y + \sum_{k = 1}^p \lambda_k
u_k$, where $y$ is a possible solution and $u_1, ..., u_p$ is a possibly empty
basis of the null solution space.
In the case of no solution, $\abe{x}$ is not needed in
\cref{eq:f-abstract}. 
Otherwise, we can compute $\abe{x_j}$ as $\{y_j + \sum_{k = 1}^m \lambda_k u_{k,
j} \mid \lambda_k \in \bools \}$.

\subsection{Trace}\label{sec:trace}
Recall that the probability of obtaining state $\rho_+$ when measuring $\rho$ is
$\ttrace{\rho_+}$.
We now describe how to compute this trace using $\grppref$ defined above.

\para{Background: Concrete Trace}
Following \cite{aaronson_improved_2004}, we compute the trace of a density matrix $\rho$ by:
\begin{equation}
	\ttrace{\rho} 
	= \sum_{i=1}^m\Re \Big( c_i \i^{\grppref(P,Q,b_i)} \Big),
	\label{eq:trace-concrete}
\end{equation}
where we define $\i^{\lightning} := 0$.
Because the trace of a density matrix is always real, $\Re(\cdot)$ is redundant,
but will be convenient to avoid complex traces in our abstraction.

\para{Abstract Trace}
For an abstract state $\rho$, we define:
\begin{equation}
	\ttrace{\abe{\rho}} 
	= r \cdot \Re\Big( \abe{c} \i^{\grppref(\abe{P},Q,\abe{b})} \Big),
	\label{eq:def-abstract-trace}
\end{equation}
where we use $\grppref(\cdot)$ as defined in \cref{eq:f-abstract}.

\section{Implementation} \label{sec:implementation}
In the following, we discuss our implementation of the abstract transformers
from \cref{sec:abstractElements} and \cref{sec:abstract-transformers} in \tool.

\para{Language and Libraries}
We implemented \tool in Python~3.8, relying on Qiskit~0.40.0~\cite{Qiskit} for
handling quantum circuits, and a combination of NumPy~1.20.0~\cite{numpy} and
Numba~0.54~\cite{numba} to handle matrix operations.

\para{Bit Encodings}
An abstract density matrix $\abe{\rho} = {r} \star
{\abe{c}}
\cdot
{\abe{P}}
\cdot 
\prod_{j=1}^n \tfrac{\I + (-1)^{\abe{b}_{j}} {Q_j}}{2}$
is encoded as a tuple $(r, \abe{c}, \abe{P}, \abe{b_{1}}, ...,\abe{b_{n}} , Q_1,
\dots, Q_n)$. To encode the concrete Pauli matrices $Q_j$, we follow concrete
stabilizer simulation encodings such as \cite{gidney_stim_2021} and encode Pauli
matrices $P$ using two bits $\encpauli(P)$ (see
\cref{sec:efficient-comp-squigg}).
To encode abstract elements of a finite set we use
bit patterns. For example, we encode $\abe{b_1} = \{1,0\} \in \abools$ as $11_2$, where the
least significant bit (i.e. the right-most bit) indicates that $0 \in \abe{b_1}$.
Analogously, we encode $\abe{v} = \{3,0\} \in \azz$ as $1001_2$.
Further, we encode $\{\Z,\Y\}$ as $1100_2$, where the indicator bits correspond
to $\Z$, $\Y$, $\X$, and $\I$, respectively, from left to right. Hence the
abstact Pauli $\abe{P} = ({\{0,3\}}, {\{\Z,\Y\}}, {\{\X\}})$ would be
represented as $(1001_2, 1100_2, 0010_2)$.

\para{Implementing Transformers}
The abstract transformers on abstract density matrices can be implemented using
operations in $\abools, \azz, \acomplex$, and $\apaulis[1]$. As $\abools, \azz$,
and $\apaulis[1]$ are small finite domains, we can implement operations in these
domains using lookup tables, which avoids the need for bit manipulation tricks.
While such tricks are applicable in our context (e.g.,
\cite{aaronson_improved_2004} uses bit manipulations to compute $\applyat[i]{H}
P \applyat[i]{H^\dagger}$ for $P \in \pauligroup{n}$), they are generally hard
to come up with~\cite{hackersdelight}.
In contrast, the efficiency of our lookup tables is comparable to that of bit
manipulation tricks, without requiring new insights for new operations.

For example, to evaluate $\{\} + \{0\}$ over $\abools$ using
\cref{eq:addition-set}, we encode the first argument $\{\}$ as $00$ and the
second argument $\{0\}$ as $01$. Looking up entry $(00,01)$ in a two-dimensional
pre-computed table then yields $00$, the encoding of the correct result $\{\}$.
We note that we cannot implement this operation directly using a XOR instruction
on encodings, as this would yield incorrect results: $00 \text{ XOR } 01 = 01
\simeq \{0\}$, which is incorrect.

\para{Gaussian Elimination}
To efficiently solve equations modulo two as discussed in
\cref{sec:abstract-transformers}, we implemented a custom Gaussian elimination
relying on bit-packing (i.e., storing $32$ boolean values in a single $32$-bit
integer).
In the future, it would be interesting to explore if Gaussian elimination could
be avoided altogether, as suggested by previous works~\cite{aaronson_improved_2004,gidney_stim_2021}.

\newcommand{\corners}[1]{\textswab{C}(#1)}

\para{Testing}
To reduce the likelihood of implementation errors, we have complemented \tool
with extensive automated tests.
We test that abstract transformers $f^\sharp$ are sound with respect to concrete
functions $f$, that is to say that
$$
\forall x_1 \in \gamma(\abe{x_1}) \cdots \forall x_k \in \gamma(\abe{x_k}). f(x_1, \dots, x_n) \in f^\sharp(\abe{x_1}, \dots, \abe{x_k}).
$$
We check this inclusion for multiple selected samples of $\abe{x_i}$ and $x_i \in \abe{x_i}$ (typically corner cases).

This approach is highly effective at catching implementation errors, which we
have found in multiple existing tools as shown in \cref{sec:eval}.
\section{Evaluation} \label{sec:eval}
We now present our evaluation of \tool, demonstrating that it can establish circuit properties no existing tool can establish.

\begin{table*}
	\caption{Description of benchmark circuits, where $\text{upper}=\{1,\dots,31\}$ and $\text{lower}=\{32,\dots,62\}$.}
	\label{tab:benchmarks}
	\centering
	\scriptsize
	\resizebox{\linewidth}{!}{
	\setlength{\tabcolsep}{1pt}
	\begin{tabular}{@{}lll@{}}
		\textbf{Circuit} & \textbf{Generation} & \textbf{Gates} (approx.) \\\hline
		\makecell[l]{
			\texttt{Cliff;Cliff} \\
			\includegraphics[scale=0.4]{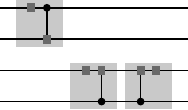}
		} & $\begin{array}{@{}l}
			c_1 \in \big( \{o(q) \mid o \in \{H,S\}, q \in \text{upper} \} \cup \{CX(q_1, q_2) \mid q_1,q_2 \in \text{upper}\} \big)^{10^4} \\
			c_2 \in \big( \{o(q) \mid o \in \{H,S\}, q \in \text{lower} \} \cup \{CX(q_1, q_2) \mid q_1,q_2 \in \text{lower}\} \big)^{10^4} \\
			\text{return } c_1; c_2; \text{opt}(c_2^\dagger)
		\end{array}$ &
		$26\text{k}\times \text{Clifford}$ \\[0.65cm]
		%%%%%%%%%%%%%%%%%%%%%%%%%%%%%%%%%%%%%%%%%%%%%%%%%%%%%%%%%%%%%%%%%%%%%%%%
		\makecell[l]{
			\texttt{Cliff+T;Cliff} \\
			\includegraphics[scale=0.4]{./figures/circuits/separate.pdf}
		} & $\begin{array}{@{}l}
			c_1 \in \big( \{o(q) \mid o \in \{H,S,T\}, q \in \text{upper} \} \cup \{CX(q_1, q_2) \mid q_1,q_2 \in \text{upper}\} \big)^{10^4} \\
			c_2 \in \big( \{o(q) \mid o \in \{H,S\}, q \in \text{lower} \} \cup \{CX(q_1, q_2) \mid q_1,q_2 \in \text{lower}\} \big)^{10^4} \\
			\text{return } c_1; c_2; \text{opt}(c_2^\dagger)
		\end{array}$ &
		\makecell[l]{
			$23\text{k}\times \text{Clifford}$, \\
			$2.5\text{k}\times T$
		 } \\[0.65cm]
		%%%%%%%%%%%%%%%%%%%%%%%%%%%%%%%%%%%%%%%%%%%%%%%%%%%%%%%%%%%%%%%%%%%%%%%%
		\makecell[l]{
			\texttt{Cliff+T;CX+T} \\
			\includegraphics[scale=0.4]{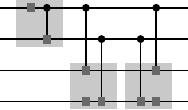}
		} & $\begin{array}{@{}l}
			c_1 \in \big( \{o(q) \mid o \in \{H,S,T\}, q \in \text{upper} \} \cup \{CX(q_1, q_2) \mid q_1,q_2 \in \text{upper}\} \big)^{10^4} \\
			c_2 \in \big( \{CX(q_1, q_2) \mid q_1 \in \text{upper}, q_2 \in \text{lower} \} \cup \{T(q) \mid q \in \text{lower} \} \big)^{10^4} \\
			\text{return } c_1; c_2; \text{opt}(c_2^\dagger)
		\end{array}$ &
		\makecell[l]{
			$18\text{k}\times \text{Clifford}$, \\
			\mbox{$9\text{k}\times T$, $40\times T^\dagger$}
		} \\[0.65cm]
		%%%%%%%%%%%%%%%%%%%%%%%%%%%%%%%%%%%%%%%%%%%%%%%%%%%%%%%%%%%%%%%%%%%%%%%%
		\makecell[l]{
			\texttt{Cliff+T;CX+T'} \\
			\includegraphics[scale=0.4]{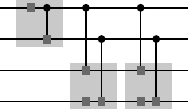}
		} & $\begin{array}{@{}l}
			c_1 \in \big( \{o(q) \mid o \in \{H,S,T\}, q \in \text{upper} \} \cup \{CX(q_1, q_2) \mid q_1,q_2 \in \text{upper}\} \big)^{10^4} \\
			c_2 \in \big( \{CX(q_1, q_2) \mid q_1 \in \text{upper}, q_2 \in \text{lower} \} \cup \{T(q) \mid q \in \text{lower} \} \big)^{10^4} \\
			\text{return } c_1; c_2; \text{opt}(c_2)
		\end{array}$ &
		\makecell[l]{
			$18\text{k}\times \text{Clifford}$, \\
			\mbox{$9\text{k}\times T$, $40\times T^\dagger$}
		} \\[0.65cm]
		%%%%%%%%%%%%%%%%%%%%%%%%%%%%%%%%%%%%%%%%%%%%%%%%%%%%%%%%%%%%%%%%%%%%%%%%
		\makecell[l]{
			\texttt{Cliff+T;H;CZ+RX} \\
			\includegraphics[scale=0.4]{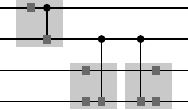}
		} & $\begin{array}{@{}l}
			c_1 \in \big( \{o(q) \mid o \in \{H,S,T\}, q \in \text{upper} \} \cup \{CX(q_1, q_2) \mid q_1,q_2 \in \text{upper}\} \big)^{10^4} \\
			c_h = H(32); \dots; H(62) \\
			c_2 \in \big( \{CZ(q_1, q_2) \mid q_1 \in \text{upper}, q_2 \in \text{lower}\} \cup \{RX_{\frac{\pi}{4}}(q) \mid q \in \text{lower} \} \big)^{10^4} \\
			\text{return } c_1; c_h; c_2; \text{opt}(c_2^\dagger); c_h
		\end{array}$ &
		\makecell[l]{
			$18\text{k}\times\text{Clifford}$, \\
			$5\text{k}\times RX_{\frac{\pi}{4}}$, \\
			$3\text{k}\times T$, $1\text{k}\times T^\dagger$
		 } \\[0.65cm]
		%%%%%%%%%%%%%%%%%%%%%%%%%%%%%%%%%%%%%%%%%%%%%%%%%%%%%%%%%%%%%%%%%%%%%%%%
		\makecell[l]{
			\texttt{Cliff+T;H;CZ+RX'} \\
			\includegraphics[scale=0.4]{./figures/circuits/cx-h.pdf}
		} & $\begin{array}{@{}l}
			c_1 \in \big( \{o(q) \mid o \in \{H,S,T\}, q \in \text{upper} \} \cup \{CX(q_1, q_2) \mid q_1,q_2 \in \text{upper}\} \big)^{10^4} \\
			c_h = H(32); \dots; H(62) \\
			c_2 \in \big( \{CZ(q_1, q_2) \mid q_1 \in \text{upper}, q_2 \in \text{lower}\} \cup \{RX_{\frac{\pi}{4}}(q) \mid q \in \text{lower} \} \big)^{10^4} \\
			\text{return } c_1; c_h; c_2; \text{opt}(c_2); c_h
		\end{array}$ &
		\makecell[l]{
			$18\text{k}\times\text{Clifford}$, \\
			$5\text{k}\times RX_{\frac{\pi}{4}}$, \\
			$4\text{k}\times T$, $40\times T^\dagger$
		 } \\[0.65cm]
		%%%%%%%%%%%%%%%%%%%%%%%%%%%%%%%%%%%%%%%%%%%%%%%%%%%%%%%%%%%%%%%%%%%%%%%%
		\makecell[l]{
			\texttt{CCX+H;Cliff} \\
			\includegraphics[scale=0.4]{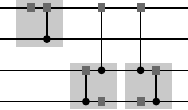}
		} & $\begin{array}{@{}l}
			c_1 \in \big( \{CCX(q_1, q_2, q_3) \mid q_1,q_2,q_3 \in \text{upper}\} \cup \{H(q) \mid q \in \text{upper} \} \big)^{10^4} \\
			c_2 \in \big( \{o(q) \mid o \in \{H,S\}, q \in \text{lower} \} \cup \\
			\hspace{2.8em} \{CX(q_1, q_2) \mid q_1 \in \text{lower}, q_2 \in \text{lower} \cup \text{upper} \} \big)^{10^4} \\
			\text{return } c_1; c_2; \text{opt}(c_2^\dagger)
		\end{array}$ &
		\makecell[l]{
			$22\text{k}\times\text{Clifford}$, \\
			\mbox{$5\text{k}\times CCX$}
		 } \\[0.65cm]
		%%%%%%%%%%%%%%%%%%%%%%%%%%%%%%%%%%%%%%%%%%%%%%%%%%%%%%%%%%%%%%%%%%%%%%%%
		\makecell[l]{
			\texttt{CCX+H;CX+T} \\
			\includegraphics[scale=0.4]{./figures/circuits/cx.pdf}
		} & $\begin{array}{@{}l}
			c_1 \in \big( \{CCX(q_1, q_2, q_3) \mid q_1,q_2,q_3 \in \text{upper}\} \cup \{H(q) \mid q \in \text{upper} \} \big)^{10^4} \\
			c_2 \in \big( \{CX(q_1, q_2) \mid q_1 \in \text{upper}, q_2 \in \text{lower} \} \cup \{T(q) \mid q \in \text{lower} \} \big)^{10^4} \\
			\text{return } c_1; c_2; \text{opt}(c_2^\dagger)
		\end{array}$ &
		\makecell[l]{
			$16\text{k}\times\text{Clifford}$, \\
			\mbox{$5\text{k}\times CCX$}, \\
			$5\text{k}\times T$, $1\text{k}\times T^\dagger$
		} \\[0.65cm]
		%%%%%%%%%%%%%%%%%%%%%%%%%%%%%%%%%%%%%%%%%%%%%%%%%%%%%%%%%%%%%%%%%%%%%%%%
		\makecell[l]{
			\texttt{CCX+H;CX+T'} \\
			\includegraphics[scale=0.4]{./figures/circuits/cx-twice.pdf}
		} & $\begin{array}{@{}l}
			c_1 \in \big( \{CCX(q_1, q_2, q_3) \mid q_1,q_2,q_3 \in \text{upper}\} \cup \{H(q) \mid q \in \text{upper} \} \big)^{10^4} \\
			c_2 \in \big( \{CX(q_1, q_2) \mid q_1 \in \text{upper}, q_2 \in \text{lower} \} \cup \{T(q) \mid q \in \text{lower} \} \big)^{10^4} \\
			\text{return } c_1; c_2; \text{opt}(c_2)
		\end{array}$ &
		\makecell[l]{
			$16\text{k}\times\text{Clifford}$, \\
			\mbox{$5\text{k}\times CCX$}, \\
			$7\text{k}\times T$, $30 \times T^\dagger$
		} \\[0.65cm]
		%%%%%%%%%%%%%%%%%%%%%%%%%%%%%%%%%%%%%%%%%%%%%%%%%%%%%%%%%%%%%%%%%%%%%%%%
		\makecell[l]{
			\texttt{RZ$_2$+H;CX} \\
			\includegraphics[scale=0.4]{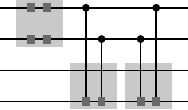}
		} & $\begin{array}{@{}l}
			c_1 \in \big( \{o(q) \mid o \in \{RZ_2, H\}, q \in \text{upper} \} \big)^{10^4} \\
			c_2 \in \big( \{CX(q_1, q_2) \mid q_1 \in \text{upper}, q_2 \in \text{lower}\} \big)^{10^4} \\
			\text{return } c_1; c_2; \text{opt}(c_2^\dagger)
		\end{array}$ &
		\makecell[l]{
			$16\text{k}\times\text{Clifford}$, \\
			\mbox{$5\text{k}\times RZ_2$}
		} \\[0.65cm]
		%%%%%%%%%%%%%%%%%%%%%%%%%%%%%%%%%%%%%%%%%%%%%%%%%%%%%%%%%%%%%%%%%%%%%%%%
		\makecell[l]{
			\texttt{RZ$_2$+H;CX'} \\
			\includegraphics[scale=0.4]{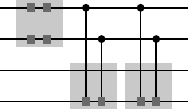}
		} & $\begin{array}{@{}l}
			c_1 \in \big( \{o(q) \mid o \in \{RZ_2, H\}, q \in \text{upper} \} \big)^{10^4} \\
			c_2 \in \big( \{CX(q_1, q_2) \mid q_1 \in \text{upper}, q_2 \in \text{lower}\} \big)^{10^4} \\
			\text{return } c_1; c_2; \text{opt}(c_2)
		\end{array}$ &
		\makecell[l]{
			$16\text{k}\times\text{Clifford}$, \\
			\mbox{$5\text{k}\times RZ_2$}
		 } \\[0.65cm]
		%%%%%%%%%%%%%%%%%%%%%%%%%%%%%%%%%%%%%%%%%%%%%%%%%%%%%%%%%%%%%%%%%%%%%%%%
		\makecell[l]{
			\texttt{MeasureGHZ} \\
			\includegraphics[scale=0.4]{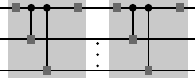}
		} & $\begin{array}{@{}l}
			c_1 = CX(1,2); \dots; CX(1,62) \\
			c_2 = H(1); c_1; \text{measure}(1), c_1 \\
			\text{return } c_2; \dots; c_2 \text{   ($100$ times)}
		\end{array}$ &
		\makecell[l]{
			$12\text{k}\times\text{Clifford}$, \\
			$100\times\text{measure}$
		}
	\end{tabular}
	}
\end{table*}

% % ugly hack from https://tex.stackexchange.com/questions/402627/using-clearpage-after-figure-in-revtex-gives-error-output-routine-didnt-u
% \makeatletter\onecolumngrid@push\makeatother
% \clearpage
% \makeatletter\onecolumngrid@pop\makeatother

\subsection{Benchmarks}
To evaluate \tool, we generated $12$ benchmark circuits, summarized and
visualized in \cref{tab:benchmarks}.

\para{Benchmark Circuit Generation}
Each circuit operates on $62$ qubits, partitioned into $31$ \emph{upper} qubits
and $31$ \emph{lower} qubits.
We picked the limit of $62$ qubits because our baseline ESS (discussed shortly)
only supports up to $63$ qubits; \tool is not subject to such a limitation.

Each circuit operates on initial state $\ket{0}$ and is constructed to ensure
that all lower qubits are eventually reverted to state $\ket{0}$. We chose this
invariant as it can be expressed for most of the evaluated tools, as we will
discuss in \cref{sec:baselines}.
Further, as some tools can only check this for one qubit at a time, we only check if the very last qubit is reverted to $\ket{0}$, instead of running $31$ independent checks (which would artificially slow down some baselines). Note that this check is of equivalent difficulty for all lower qubits.

\para{Benchmark Details}
\cref{tab:benchmarks} details how each benchmark circuit was generated.
Most of the circuits are built from three concatenated subcircuits. First, $c_1$
modifies the upper qubits, then $c_2$ modifies the lower qubits (potentially
using gates controlled by the upper qubits) and finally $c_3$ reverts all lower qubits to $\ket{0}$, but in a non-trivial way.
Circuit \texttt{CCX+H;Cliff} slightly deviates from this pattern, as it also
modifies the upper qubits using gates controller by lower qubits.
Further, circuits \texttt{Cliff+T;H;CZ+RX} and \texttt{Cliff+T;H;CZ+RX'}
additionally apply two layers of $H$ gates to the lower qubits.
Finally, circuit \texttt{MeasureGHZ} applies internal measurements, as discussed
below.

The majority of circuits revert the lower qubits to $\ket{0}$ by applying $c_3$,
the inverse of $c_2$ but optimized using PyZX~\cite{kissinger2020Pyzx}---this
obfuscates the fact that $c_2$ and $c_3$ cancel out.
Four circuits, marked with a trailing prime (\texttt{'}), generate $c_3$ by
optimizing the un-inverted $c_2$. They still reset all lower qubits to
$\ket{0}$, but establishing this requires advanced reasoning. Specifically,
\texttt{RZ$_2$+H;CX'} flips each lower qubit an even number of
times.~\footnote{More precisely, when representing the quantum state as a sum
over computational basis states, an even number of flips are applied to each
qubit of each summand.}
Similarly, \texttt{Cliff+T;CX+T'} and \texttt{CCX+H;CX+T'} additionally modify
the phase but still flip each lower qubit an even number of times. Finally,
\texttt{Cliff+T;H;CZ+RX'} flips between states $\ket{+}$ and $\ket{-}$ an even
number of times, where $RX_{\tfrac{\pi}{4}}$ only modifies the phase.

The last benchmark \texttt{MeasureGHZ} first generates a $GHZ$ state
$\tfrac{1}{\sqrt{2}} \ket{0\cdots0} + \tfrac{1}{\sqrt{2}} \ket{1\cdots1}$, and
collapses it to $\ket{0\cdots0}$ or $\ket{1\cdots1}$ by measuring the first
qubit. Then, it resets all qubits to $\ket{0}$ except for the first one. It then
repeats this process, with the first qubit starting in either $\ket{0}$ or
$\ket{1}$. Thus, the state before measurement is either $\tfrac{1}{\sqrt{2}}
\ket{0\cdots0} + \tfrac{1}{\sqrt{2}} \ket{1\cdots1}$ or $\tfrac{1}{\sqrt{2}}
\ket{0\cdots0} - \tfrac{1}{\sqrt{2}} \ket{1\cdots1}$, but every repetition still
resets all lower qubits to $\ket{0}$.

\para{Discussion}
Our benchmark covers a wide variety of gates, with all applying Clifford gates,
seven applying $T$ gates, three applying $CCX$ gates, two applying
$RX_{\frac{\pi}{4}}$ gates (one qubit gate, rotation around the $X$ axis of
$\tfrac{\pi}{4}$ radians), and two applying $RZ_2$ gates (one qubit gate,
rotation around the $Z$ axis of $2$ radians).

All benchmarks are constructed to revert the lower qubits to $\ket{0}$, but in a
non-obvious way. As fully precise simulation of most benchmarks is unrealistic,
we expect that over-approximation is typically necessary to establish this fact.

\subsection{Baselines} \label{sec:baselines}
We now discuss how we instantiated existing tools to establish that a circuit
$c$ evolves a qubit $q$ to state $\ket{0}$. Overall, we considered two tools
based on stabilizer simulation (ESS~\cite{bravyi_simulation_2019} and
QuiZX~\cite{kissinger_simulating_2022}), one tool based on the Feynman path
integral (Feynman~\cite{amy_towards_2019}), one tool based on abstract
interpretation (YP21~\cite{yu_quantum_2021}, in two different modes), and one
tool based on state vectors (Statevector as implemented by
Qiskit~\cite{Qiskit}).

\para{ESS}
Qiskit~\cite{Qiskit} provides an extended stabilizer simulator 
implementing the ideas published in \cite{bravyi_simulation_2019} which
(i)~decomposes quantum circuits into Clifford circuits, (ii)~simulates these
circuits separately, and (iii)~performs measurements by an aggregation across
these circuits.
To check if a circuit $c$ consistently evolves a qubit $q$ to $\ket{0}$, we check if $c$ extended by a measurement of $q$ always yields $0$.
To run our simulation, we used default parameters.

\para{QuiZX}
QuiZX~\cite{kissinger_simulating_2022} improves upon
\cite{bravyi_simulation_2019} by alternating between decomposing circuits
(splitting non-Clifford gates into Clifford gates) and optimizing the decomposed
circuits (which may further reduce non-Clifford gates).
We can use QuiZX to establish that a qubit is in state $\ket{0}$ by "plugging"
output $q$ as $\ket{1}$ and establishing that the probability of this output is
zero.~\footnote{The use of plugging is described on
\url{https://github.com/Quantomatic/quizx/issues/9}.}

\para{Feynman}
Feynman~\cite{amy_towards_2019} allows to verify quantum circuits based on the
Feynman path integral.
Its implementation\footnote{Tool available at
\url{https://github.com/meamy/feynman}} supports two main use cases, namely
optimization and checking the equivalence of two circuits.
While these use cases cannot prove that a circuit resets a qubit to $\ket{0}$,
we can use Feynman's equivalence check to check whether the circuits in
\cref{tab:benchmarks} are equivalent to a simplified version which
performs no operation at all on lower qubits.
We check this equivalence for all circuits, even for those where we know it does
not hold (namely all whose name ends with a prime), allowing us to  confirm that
Feynman cannot scale to any of our benchmarks (see
\cref{sec:eval-results}).

We note that Feynman currently does not support internal
measurements.~\footnote{\url{https://github.com/meamy/feynman/issues/8}}

\para{YP21}
Like \tool, \baseline~\cite{yu_quantum_2021} also uses abstract interpretation,
but relies on projectors instead of stabilizer simulation. Specifically, it
encodes the abstract state of selected (small) subsets of qubits as projectors
$\{P_j\}_{j \in \mathcal{J}}$, which constrain the state of these qubits to the
range of $P_j$.

To check if a qubit $q$ is in state $\ket{0}$, we check if the subspace
resulting from intersecting the range of all $P_j$ is a subset of the range of
$\I + \applyat[q]{Z}$---an operation which is natively supported by~\baseline.

When running \baseline, we used the two execution modes suggested in its
original evaluation~\cite{yu_quantum_2021}.
The first mode tracks the state of all pairs of qubits, while the second
considers subsets of $5$ qubits that satisfy a particular condition (for
details, see \cite[\secref{9}]{yu_quantum_2021}).
Because \cite{yu_quantum_2021} does not discuss which execution mode to pick for
new circuits, we evaluated all circuits in both modes.

We note that because \baseline does not support $CX(a,b)$ for $a>b$, we instead
encoded such gates as $H(a);H(b); CX(b,a); H(b); H(a)$.

\para{Statevector}
Qiskit~\cite{Qiskit} further provides a simulator based on state vectors, which
we also used for completeness.

\para{Abstraqt}
In \tool, we can establish that a qubit is in state $\ket{0}$ by measuring the
final abstract state $\abe{\rho}$ in basis $\applyat[i]{Z}$ and checking if the
probability of obtaining $\ket{1}$ is $0$.

% srlx
% htop
% lscpu
% lsb_release -a

\para{Experimental Setup}
We executed all experiments on a machine with $110$ GB RAM and $56$ cores
at $2.6$~GHz, running Ubuntu 22.04.
Because some tools consumed excessive amounts of memory, we limited them to $12$
GB of RAM.
This was not necessary for \tool, which never required more than $600$ MB of RAM.
We limited each tool to a single thread.

\subsection{Results} \label{sec:eval-results}

\cref{tab:results} summarizes the results when using all tools
discussed in \cref{sec:baselines} to establish that the last qubit in
$10$ randomly selected instantiations of each benchmark from
\cref{tab:benchmarks} is in state $\ket{0}$.
Overall, it demonstrates that while \tool can establish this for all benchmarks
within minutes, QuiZX can only establish it for a few instances, and all other
tools cannot establish it for any benchmark. Further, we found that for some
circuits the established simulation tool ESS yields incorrect results. We now
discuss the results of each tool in more details.

\para{MeasureGHZ}
Importantly, no baseline tool except \tool can simultaneously simulate both
outcomes of a measurement, without incurring an exponential blow-up. Therefore,
for \texttt{MeasureGHZ}, we consider internal measurements as an unsupported
operation in these tools.
We note that we could randomly select one measurement outcome and simulate the
remainder of the circuit for it, but then we can only establish that the final
state is $\ket{0}$ \emph{for a given sequence of measurement outcomes}. In
contrast, a single run of \tool can establish that the final state is $\ket{0}$
for all possible measurement outcomes (see also \cref{sec:measurement}).

\begin{table*}
	\caption{Success rates when running simulators on benchmarks from \cref{tab:benchmarks}.}
	\label{tab:results}
	\footnotesize
	\centering
	\resizebox{\linewidth}{!}{
	\setlength{\tabcolsep}{3pt}
	\begin{tabular}{@{}llllllll@{}}
		\textbf{Label} &
		\textbf{Abstraqt} &
		\textbf{QuiZX} &
		\textbf{ESS} &
		\textbf{Feynman} &
		\textbf{YP21 (mode 1)} &
		\textbf{YP21 (mode 2)} &
		\textbf{Statevec.}\\\hline
		%%%%%%%%%%%%%%%%%%
		\texttt{Cliff;Cliff} &
		$\textbf{100\%}$ & % abstraqt
		$\phantom{00}0\%$ (E) & % QuiZX
		$0\%$ (I) & % ESS
		$0\%$ (T,M) & % Feynman
		$0\%$ (T,P) & % YP1
		$0\%$ (I) & % YP2
		$0\%$ (M) \\ % Statevector
		%%%%%%%%%%%%%%%%%%
		\texttt{Cliff+T;Cliff} &
		$\textbf{100\%}$ & % abstraqt
		$\phantom{0}70\%$ (T) & % QuiZX
		$0\%$ (I) & % ESS
		$0\%$ (T,M) & % Feynman
		$0\%$ (T,P) & % YP1
		$0\%$ (I) & % YP2
		$0\%$ (M) \\ % Statevector
		%%%%%%%%%%%%%%%%%%
		\texttt{Cliff+T;CX+T} &
		$\textbf{100\%}$ & % abstraqt
		$\phantom{0}80\%$ (M) & % QuiZX
		$0\%$ (M) & % ESS
		$0\%$ (T) & % Feynman
		$0\%$ (T,E,P) & % YP1
		$0\%$ (I) & % YP2
		$0\%$ (M) \\ % Statevector
		%%%%%%%%%%%%%%%%%%
		\texttt{Cliff+T;CX+T'} &
		$\textbf{100\%}$ & % abstraqt
		$\phantom{00}0\%$ (M) & % QuiZX
		$0\%$ (M) & % ESS
		$0\%$ (T) & % Feynman
		$0\%$ (T,E,P) & % YP1
		$0\%$ (I) & % YP2
		$0\%$ (M) \\ % Statevector
		%%%%%%%%%%%%%%%%%%
		\texttt{Cliff+T;H+CZ+RX} &
		$\textbf{100\%}$ & % abstraqt
		$\phantom{0}60\%$ (M) & % QuiZX
		$0\%$ (M) & % ESS
		$0\%$ (T) & % Feynman
		$0\%$ (T,P) & % YP1
		$0\%$ (I) & % YP2
		$0\%$ (M) \\ % Statevector
		%%%%%%%%%%%%%%%%%%
		\texttt{Cliff+T;H+CZ+RX'} &
		$\textbf{100\%}$ & % abstraqt
		$\phantom{00}0\%$ (T,M) & % QuiZX
		$0\%$ (M) & % ESS
		$0\%$ (T) & % Feynman
		$0\%$ (T,P) & % YP1
		$0\%$ (I) & % YP2
		$0\%$ (M) \\ % Statevector
		%%%%%%%%%%%%%%%%%%
		\texttt{CCX+H;Cliff} &
		$\textbf{100\%}$ & % abstraqt
		$\phantom{00}0\%$ (T) & % QuiZX
		$0\%$ (M) & % ESS
		$0\%$ (M) & % Feynman
		$0\%$ (T) & % YP1
		$0\%$ (T) & % YP2
		$0\%$ (M) \\ % Statevector
		%%%%%%%%%%%%%%%%%%
		\texttt{CCX+H;CX+T} &
		$\textbf{100\%}$ & % abstraqt
		$\phantom{0}50\%$ (T) & % QuiZX
		$0\%$ (M) & % ESS
		$0\%$ (T) & % Feynman
		$0\%$ (T) & % YP1
		$0\%$ (T) & % YP2
		$0\%$ (M) \\ % Statevector
		%%%%%%%%%%%%%%%%%%
		\texttt{CCX+H;CX+T'} &
		$\textbf{100\%}$ & % abstraqt
		$\phantom{00}0\%$ (T,M) & % QuiZX
		$0\%$ (M) & % ESS
		$0\%$ (T) & % Feynman
		$0\%$ (T) & % YP1
		$0\%$ (T) & % YP2
		$0\%$ (M) \\ % Statevector
		%%%%%%%%%%%%%%%%%%
		\texttt{RZ$_2$+H;CX} &
		$\textbf{100\%}$ & % abstraqt
		$\phantom{00}0\%$ (E) & % QuiZX
		$0\%$ (T,M) & % ESS
		$0\%$ (U) & % Feynman
		$0\%$ (U) & % YP1
		$0\%$ (U) & % YP2
		$0\%$ (M) \\ % Statevector
		%%%%%%%%%%%%%%%%%%
		\texttt{RZ$_2$+H;CX'} &
		$\textbf{100\%}$ & % abstraqt
		$\phantom{00}0\%$ (E) & % QuiZX
		$0\%$ (M) & % ESS
		$0\%$ (U) & % Feynman
		$0\%$ (U) & % YP1
		$0\%$ (U) & % YP2
		$0\%$ (M) \\ % Statevector
		%%%%%%%%%%%%%%%%%%
		\texttt{MeasureGHZ} &
		$\textbf{100\%}$ & % abstraqt
		$\phantom{00}0\%$ (U) & % QuiZX
		$0\%$ (U) & % ESS
		$0\%$ (U) & % Feynman
		$0\%$ (U) & % YP1
		$0\%$ (U) & % YP2
		$0\%$ (U) \\ % Statevector
		%%%%%%%%%%%%%%%%%%
		\hline
		\textbf{Overall success} &
		$\textbf{100\%}$ & % abstraqt
		$\phantom{0}22\%$ & % QuiZX
		$0\%$ & % ESS
		$0\%$ & % Feynman
		$0\%$ & % YP1
		$0\%$ & % YP2
		$0\%$ % Statevector
		\\[1em]
		%%%%%%%%%%%%%%%%%%
		\multicolumn{8}{@{}p{\linewidth}}{%
			T: timeout (6h), M: out of memory, U: unsupported operation in the circuit,
		} \\
		\multicolumn{8}{@{}p{\linewidth}}{%
			I: incorrect simulation results, P: too imprecise, E: internal error
		}
	\end{tabular}
	}
\end{table*}

\begin{table*}
	\caption{Detailed comparison of outcomes from \tool and QuiZX, including runtimes of successful runs.}
	\label{tab:quizx}
	\centering
	\resizebox{\linewidth}{!}{
	\begin{tabular}{%
		@{}%
		l%
		|%
		*{3}l%
		|%
		l@{\hspace{1pt}}%
		l@{\hspace{1pt}}%
		l*{2}%
		l@{}%
	}
		\footnotesize
		\textbf{Label} &
		\multicolumn{3}{@{}c|}{\textbf{\tool}} &
		\multicolumn{5}{@{}c}{\textbf{QuiZX}}
		\\
		& % label
		Outcomes &
		min [s] &
		max [s] &
		\multicolumn{3}{@{}c}{Outcomes} &
		min [s] &
		max [s] \\
		\hline
		%%%%%%%%%%%%%%%%%%
		\texttt{Cliff;Cliff} &
		$10 \times \text{\xyes{}}$ &
		$\phantom{0}24$ &
		$\phantom{0}33$ &
		$0 \times \text{\xyes{}}$, & & $10 \times \text{E}$ &
		- &
		- \\
		%%%%%%%%%%%%%%%%%
		\texttt{Cliff+T;Cliff} &
		$10 \times \text{\xyes{}}$ &
		$\phantom{0}32$ &
		$\phantom{0}47$ &
		$7 \times \text{\xyes{}}$, & $\phantom{0}3 \times \text{T}$ & &
		$5.5 \cdot 10^3$ &
		$2.0 \cdot 10^4$ \\
		%%%%%%%%%%%%%%%%%%
		\texttt{Cliff+T;CX+T} &
		$10 \times \text{\xyes{}}$ &
		$\phantom{0}46$ &
		$\phantom{0}63$ &
		$8 \times \text{\xyes{}}$, & $\phantom{0}2 \times \text{M}$ & &
		$2.0 \cdot 10^3$ &
		$9.4 \cdot 10^3$
		\\
		%%%%%%%%%%%%%%%%%%
		\texttt{Cliff+T;CX+T'} &
		$10 \times \text{\xyes{}}$ &
		$\phantom{0}47$ &
		$\phantom{0}65$ &
		$0 \times \text{\xyes{}}$, & $10 \times \text{T}$ & &
		- &
		-
		\\
		%%%%%%%%%%%%%%%%%%
		\texttt{Cliff+T;H+CZ+RX} &
		$10 \times \text{\xyes{}}$ &
		$\phantom{0}58$ &
		$\phantom{0}69$ &
		$6 \times \text{\xyes{}}$, & & $\phantom{0}4 \times \text{M}$ &
		$3.6 \cdot 10^3$ &
		$1.4 \cdot 10^4$
		\\
		%%%%%%%%%%%%%%%%%%
		\texttt{Cliff+T;H+CZ+RX'} &
		$10 \times \text{\xyes{}}$ &
		$\phantom{0}52$ &
		$\phantom{0}71$ &
		$0 \times \text{\xyes{}}$, & $\phantom{0}1 \times \text{T}$, & $\phantom{0}9 \times \text{M}$ &
		- &
		-
		\\
		%%%%%%%%%%%%%%%%%%
		\texttt{CCX+H;Cliff} &
		$10 \times \text{\xyes{}}$ &
		$143$ &
		$155$ &
		$0 \times \text{\xyes{}}$, & $10 \times \text{T}$ & &
		- &
		-
		\\
		%%%%%%%%%%%%%%%%%%
		\texttt{CCX+H;CX+T} &
		$10 \times \text{\xyes{}}$ &
		$155$ &
		$173$ &
		$5 \times \text{\xyes{}}$, & $\phantom{0}5 \times \text{T}$ & &
		$5.9 \cdot 10^3$ &
		$7.9 \cdot 10^3$
		\\
		%%%%%%%%%%%%%%%%%%
		\texttt{CCX+H;CX+T'} &
		$10 \times \text{\xyes{}}$ &
		$155$ &
		$173$ &
		$0 \times \text{\xyes{}}$, & $\phantom{0}1 \times \text{T}$, & $\phantom{0}9 \times \text{M}$ &
		- &
		-
		\\
		%%%%%%%%%%%%%%%%%%
		\texttt{RZ$_2$+H;CX} &
		$10 \times \text{\xyes{}}$ &
		$\phantom{0}37$ &
		$\phantom{0}47$ &
		$0 \times \text{\xyes{}}$, & & $10 \times \text{E}$ &
		- &
		-
		\\
		%%%%%%%%%%%%%%%%%%
		\texttt{RZ$_2$+H;CX'} &
		$10 \times \text{\xyes{}}$ &
		$\phantom{0}37$ &
		$\phantom{0}46$ &
		$0 \times \text{\xyes{}}$, & & $10 \times \text{E}$ &
		- &
		-
		\\
		%%%%%%%%%%%%%%%%%%
		\texttt{MeasureGHZ} &
		$10 \times \text{\xyes{}}$ &
		$\phantom{0}23$ &
		$\phantom{0}32$ &
		$0 \times \text{\xyes{}}$, & & $10 \times \text{U}$ &
		- &
		- \\[1em]
		\multicolumn{9}{@{}p{\linewidth}}{%
			T: timeout (6h), M: out of memory, U: unsupported operation in the circuit, E: internal error
		}
	\end{tabular}
	}
\end{table*}

\para{QuiZX}
As QuiZX is the only baseline tool solving some of our benchmark instances, we provide a detailed comparison to it in \cref{tab:quizx}.

Overall, QuiZX cannot consistently handle any of the benchmarks from
\cref{tab:benchmarks}, 
Instead, it often either times out or runs out of memory.
Further, QuiZX consistently runs into an internal error when simulating
\texttt{RZ$_2$+H;CX} and \texttt{RZ$_2$+H;CX'}.
Surprisingly, QuiZX also consistently fails to simulate \texttt{Cliff;Cliff},
which we conjecture is due to a bug for circuits that do not contain
non-Clifford gates. After adding a single $T$ gate, simulation is successful.

Importantly, even when QuiZX succeeds, it is significantly slower than \tool,
sometimes by more than two orders of magnitude.

\para{ESS}
Surprisingly, ESS simulates circuits \texttt{Cliff;Cliff} and
\texttt{Cliff+T;Cliff} incorrectly. Specifically, it samples the impossible
measurement of $1$ around 50\% of cases. Interestingly, smaller circuits
generated with the same process are handled correctly. It is reassuring to see
that \tool allows us to discover such instabilities in established tools.

It may be surprising that ESS returns an incorrect result for
\texttt{Cliff+T;Cliff} instead of timing out, although the circuit contains many
$T$ gates---this is because Qiskit can establish that the Clifford+T part of the
circuit is irrelevant when measuring the last qubit. For all remaining circuits,
ESS runs out of memory or times out, as it decomposes the circuit into
exponentially many Clifford circuits.

\para{Feynman}
Feynman consistently either times out, runs out of memory, or does not support a
relevant operation (namely measurement and $RZ_2$).

\para{\baseline}
\baseline typically either times out, throws an internal error, does not support
a relevant operation (e.g., measurements or $RZ_2$), or returns incorrect
results. The latter is because on some circuits, mode 2 choses an empty set of
projectors, which leads to trivially unsound results. When \baseline does
terminate, it is too imprecise to establish that the last qubit is in state
$\ket{0}$.

\para{Statevector}
Unsurprisingly, statevector simulation cannot handle the circuits in
\cref{tab:results}. This is because it requires space exponential in the number
of qubits, which precludes simulating any of the benchmarks.

\subsection{Limitations and Discussion} \label{sec:limitations}
We note that our benchmarks are designed to showcase successful applications of
\tool where it outperforms existing tools. Of course, \tool is not precise on
all circuits---e.g., \tool quickly loses precision on general Clifford+T
circuits (analogously to the imprecise measurement discussed in
\cref{sec:overview}).

\para{Future Abstractions}
We expect that for many real-world circuits, existing approaches work better
than the current implementation of \tool. However, as \tool only abstracts the
first stabilizer simulation generalized to non-Clifford
gates~\cite[\secref{VII-C}]{aaronson_improved_2004}, we believe it paves the way
to also abstract more recent stabilizer simulators.
For example, ESS~\cite{bravyi_simulation_2019} operates on so-called
\emph{CH-forms} which, like the generalized stabilizer simulation underlying
\tool, can be encoded using bits and complex numbers. Hence, it seems plausible
that our ideas could be adapted to abstract ESS. QuiZX operates on
\emph{ZX-diagrams} consisting of graphs whose nodes are parametrized by rotation
angles $\alpha$. Again, a promising direction for future research is introducing
abstract ZX-diagrams that support abstract rotation angles.
This is particularly promising because both ESS and QuiZX scale better
in number of $T$ gates than \cite[\secref{VII-C}]{aaronson_improved_2004}: with
$2^n$ instead of $4^n$.

We note however that not all concrete simulation techniques are directly
amenable to abstraction. For example, when naively abstracting the Clifford
simulation by Aaronson and Gottesmann, applying a measurement requires selecting
an entry in an boolean matrix that definitively equals one~\cite[Case I in
\secref{III}]{aaronson_improved_2004}---it is unclear how to generalize this to
abstract boolean matrices whose entries may be $\{0,1\}$.

\para{Improving \tool}
Another promising route towards better abstractions in incrementally improving
\tool itself. For example, it would be interesting to consider the effect of
keeping more than one abstract summand, abstracting $P_i$ or $b_{ij}$ using a
custom relational domain (which retains information about the relationship
between different values)~\cite{mine_weakly_2004}, or a more precise abstraction
for complex numbers by taking into account that restricted gate sets such as
Clifford+T only induce matrices over finite sets of values.

\para{Summary}
Overall, we believe that all tools in \cref{tab:results} are valuable
to analyze quantum circuits.
We are hoping that addressing some limitations of the considered baselines
(e.g., fixing bugs in QuiZX and ESS) and cross-pollinating ideas (e.g.,
extending QuiZX by abstract interpretation) will allow the community to benefit
from the fundamentally different mathematical foundations of all tools.

\section{Related Work} \label{sec:related}
Here, we discuss works related to the goal and methods of \tool.

\para{Quantum Abstract Interpretation}
Some existing works have investigated abstract interpretation for simulating
quantum
circuits~\cite{yu_quantum_2021,perdrix_quantum_2008,honda_analysis_2015}.
As \cite{yu_quantum_2021} is not specialized for Clifford circuits, it is very
imprecise on the circuits investigated in \cref{sec:eval}: it cannot derive that
the lower qubits are $\ket{0}$ for any of them.
While \cite{perdrix_quantum_2008,honda_analysis_2015} are inspired by stabilizer
simulation, they only focus on determining if certain qubits are entangled or
not, whereas \tool can extract more precise information about the state.
Further, both tools are inherently imprecise on non-Clifford gates---in
contrast, a straight-forward extension of \tool can treat some non-Clifford
gates precisely at the exponential cost of not merging summands.

\para{Stabilizer Simulation}
The Gottesman-Knill theorem~\cite{gottesmanHeisenberg1998} established that
stabilizers can be used to efficiently simulate Clifford circuits.
Stim~\cite{gidney_stim_2021} is a recent implementation of such a simulator,
which only supports Clifford gates and Pauli measurements.

Stabilizer simulation was extended to allow for non-Clifford gates at an
exponential cost, while still allowing efficient simulation of Clifford
gates~\cite[\secref{VII-C}]{aaronson_improved_2004}.
Various works build upon this insight, handling Clifford gates efficiently but
suffering from an exponential blow-up on non-Clifford
gates~\cite{gottesmantypes,kissinger_simulating_2022,bravyi_simulation_2019,pashayan_fast_2022,kissinger_classical_2022}.
In our evaluation, we demonstrate that \tool extends the reach of
state-of-the-art stabilizer simulation by comparing to two tools from this
category, ESS~\cite{bravyi_simulation_2019} (chosen because it is implemented in
the popular Qiskit library) and QuiZX~\cite{kissinger_simulating_2022} (chosen
because it is a recent tool reporting favorable runtimes). 

\para{Verifying Quantum Programs}
Another approach to establishing circuit properties is end-to-end formal program
verification, as developed in \cite{hietala-verif} for instance. However, this
approach often requires new insights for each program it is applied to. Even
though recent works have greatly improved verification automation, proving even
the simplest programs still requires a significant time
investment~\cite{chareton2020deductive}, whereas our approach can analyze it
without any human time investment.

The work~\cite{invariants_expo_2017} automatically generates rich invariants,
but is exponential in the number of qubits, limiting its use to small circuits.
Finally, \cite{amy_towards_2019} can automatically verify the equivalence of two
given circuits, but times out on the benchmarks considered in
\cref{sec:eval}.
\section{Conclusion} \label{sec:conclusion}
In this work, we have demonstrated that combining abstract interpretation with
stabilizer simulation allows to establish circuit properties that are
intractable otherwise.

Our key idea was to over-approximate the behavior of non-Clifford gates in the
generalized stabilizer simulation of Aaronson and
Gottesman~\cite{aaronson_improved_2004} by merging summands in the sum
representation of the quantum states density matrix.
Our carefully chosen abstract domain allows us to define efficient abstract transformers that approximate each of the concrete stabilizer simulation functions, including measurement.

%%%%%%%%%%%%%%
% BODY - END %
%%%%%%%%%%%%%%
% this message marks the end of the body (used to check if we are over the page
% limit)
\message{^^JLASTBODYPAGE \thepage^^J}

%%%%%%%%%%%%%%%%
% BIBLIOGRAPHY %
%%%%%%%%%%%%%%%%
% may be replaced by new bibliography
% \pagebreak
\bibliography{references}

\begin{thebibliography}{10}

\bibitem{gottesmanHeisenberg1998}
Daniel Gottesman.
\newblock ``The {Heisenberg} {Representation} of {Quantum} {Computers}''.
\newblock \href{https://dx.doi.org/10.48550/arXiv.quant-ph/9807006}{Technical
  Report arXiv:quant-ph/9807006}.
\newblock arXiv~(1998).

\bibitem{aaronson_improved_2004}
Scott Aaronson and Daniel Gottesman.
\newblock ``Improved {Simulation} of {Stabilizer} {Circuits}''.
\newblock \href{https://dx.doi.org/10.1103/PhysRevA.70.052328}{Physical Review
  A {\bf 70}, 052328}~(2004).

\bibitem{gottesmantypes}
Robert Rand, Aarthi Sundaram, Kartik Singhal, and Brad Lackey.
\newblock ``Extending gottesman types beyond the clifford group''.
\newblock In The Second International Workshop on Programming Languages for
  Quantum Computing (PLanQC 2021).
\newblock ~(2021).
\newblock
  url:~\url{https://pldi21.sigplan.org/details/planqc-2021-papers/9/Extending-Gottesman-Types-Beyond-the-Clifford-Group}.

\bibitem{kissinger_simulating_2022}
Aleks Kissinger and John van~de Wetering.
\newblock ``Simulating quantum circuits with {ZX}-calculus reduced stabiliser
  decompositions''.
\newblock \href{https://dx.doi.org/10.1088/2058-9565/ac5d20}{Quantum Science
  and Technology {\bf 7}, 044001}~(2022).

\bibitem{bravyi_simulation_2019}
Sergey Bravyi, Dan Browne, Padraic Calpin, Earl Campbell, David Gosset, and
  Mark Howard.
\newblock ``Simulation of quantum circuits by low-rank stabilizer
  decompositions''.
\newblock \href{https://dx.doi.org/10.22331/q-2019-09-02-181}{Quantum {\bf 3},
  181}~(2019).

\bibitem{pashayan_fast_2022}
Hakop Pashayan, Oliver Reardon-Smith, Kamil Korzekwa, and Stephen~D. Bartlett.
\newblock ``Fast estimation of outcome probabilities for quantum circuits''.
\newblock \href{https://dx.doi.org/10.1103/PRXQuantum.3.020361}{PRX Quantum
  {\bf 3}, 020361}~(2022).

\bibitem{kissinger_classical_2022}
\href{https://dx.doi.org/10.4230/LIPICS.TQC.2022.5}{``Classical simulation of
  quantum circuits with partial and graphical stabiliser decompositions''}.
\newblock Schloss Dagstuhl - Leibniz-Zentrum für Informatik~(2022).

\bibitem{cousotAbstract1977}
Patrick Cousot and Radhia Cousot.
\newblock ``Abstract {Interpretation}: {A} {Unified} {Lattice} {Model} for
  {Static} {Analysis} of {Programs} by {Construction} or {Approximation} of
  {Fixpoints}''.
\newblock In Proceedings of the 4th {ACM} {SIGACT}-{SIGPLAN} {Symposium} on
  {Principles} of {Programming} {Languages}.
\newblock \href{https://dx.doi.org/10.1145/512950.512973}{Pages 238--252}.
\newblock {POPL} '77New York, NY, USA~(1977). ACM.

\bibitem{cousot1992abstract}
Patrick Cousot and Radhia Cousot.
\newblock ``Abstract interpretation frameworks''.
\newblock \href{https://dx.doi.org/10.1093/logcom/2.4.511}{Journal of logic and
  computation {\bf 2}, 511--547}~(1992).

\bibitem{blanchet_static_2003}
Bruno Blanchet, Patrick Cousot, Radhia Cousot, Jérome Feret, Laurent
  Mauborgne, Antoine Miné, David Monniaux, and Xavier Rival.
\newblock ``A static analyzer for large safety-critical software''.
\newblock \href{https://dx.doi.org/10.1145/780822.781153}{ACM SIGPLAN Notices
  {\bf 38}, 196--207}~(2003).

\bibitem{logozzo_pentagons_2010}
Francesco Logozzo and Manuel Fähndrich.
\newblock ``Pentagons: {A} weakly relational abstract domain for the efficient
  validation of array accesses''.
\newblock \href{https://dx.doi.org/10.1016/j.scico.2009.04.004}{Science of
  Computer Programming {\bf 75}, 796--807}~(2010).

\bibitem{gehr_ai2_2018}
Timon Gehr, Matthew Mirman, Dana Drachsler-Cohen, Petar Tsankov, Swarat
  Chaudhuri, and Martin Vechev.
\newblock ``{AI2}: {Safety} and {Robustness} {Certification} of {Neural}
  {Networks} with {Abstract} {Interpretation}''.
\newblock In 2018 {IEEE} {Symposium} on {Security} and {Privacy} ({SP}).
\newblock \href{https://dx.doi.org/10.1109/SP.2018.00058}{Pages 3--18}.
\newblock San Francisco, CA~(2018). IEEE.

\bibitem{nielsen_quantum_2010}
Michael~A. Nielsen and Isaac~L. Chuang.
\newblock ``Quantum computation and quantum information: 10th anniversary
  edition''.
\newblock \href{https://dx.doi.org/10.1017/CBO9780511976667}{Cambridge
  University Press}. ~(2010).

\bibitem{Qiskit}
Gadi Aleksandrowicz, Thomas Alexander, Panagiotis Barkoutsos, Luciano Bello,
  Yael Ben-Haim, David Bucher, Francisco~Jose Cabrera-Hernández, Jorge
  Carballo-Franquis, Adrian Chen, Chun-Fu Chen, Jerry~M. Chow, Antonio~D.
  Córcoles-Gonzales, Abigail~J. Cross, Andrew Cross, Juan Cruz-Benito, Chris
  Culver, Salvador De La~Puente González, Enrique De~La Torre, Delton Ding,
  Eugene Dumitrescu, Ivan Duran, Pieter Eendebak, Mark Everitt, Ismael~Faro
  Sertage, Albert Frisch, Andreas Fuhrer, Jay Gambetta, Borja~Godoy Gago, Juan
  Gomez-Mosquera, Donny Greenberg, Ikko Hamamura, Vojtech Havlicek, Joe
  Hellmers, Łukasz Herok, Hiroshi Horii, Shaohan Hu, Takashi Imamichi,
  Toshinari Itoko, Ali Javadi-Abhari, Naoki Kanazawa, Anton Karazeev, Kevin
  Krsulich, Peng Liu, Yang Luh, Yunho Maeng, Manoel Marques, Francisco~Jose
  Martín-Fernández, Douglas~T. McClure, David McKay, Srujan Meesala, Antonio
  Mezzacapo, Nikolaj Moll, Diego~Moreda Rodríguez, Giacomo Nannicini, Paul
  Nation, Pauline Ollitrault, Lee~James O'Riordan, Hanhee Paik, Jesús Pérez,
  Anna Phan, Marco Pistoia, Viktor Prutyanov, Max Reuter, Julia Rice,
  Abdón~Rodríguez Davila, Raymond Harry~Putra Rudy, Mingi Ryu, Ninad Sathaye,
  Chris Schnabel, Eddie Schoute, Kanav Setia, Yunong Shi, Adenilton Silva,
  Yukio Siraichi, Seyon Sivarajah, John~A. Smolin, Mathias Soeken, Hitomi
  Takahashi, Ivano Tavernelli, Charles Taylor, Pete Taylour, Kenso Trabing,
  Matthew Treinish, Wes Turner, Desiree Vogt-Lee, Christophe Vuillot,
  Jonathan~A. Wildstrom, Jessica Wilson, Erick Winston, Christopher Wood,
  Stephen Wood, Stefan Wörner, Ismail~Yunus Akhalwaya, and Christa Zoufal.
\newblock ``Qiskit: An open-source framework for quantum computing''~(2019).

\bibitem{numpy}
Charles~R. Harris, K.~Jarrod Millman, St{\'{e}}fan~J. van~der Walt, Ralf
  Gommers, Pauli Virtanen, David Cournapeau, Eric Wieser, Julian Taylor,
  Sebastian Berg, Nathaniel~J. Smith, Robert Kern, Matti Picus, Stephan Hoyer,
  Marten~H. van Kerkwijk, Matthew Brett, Allan Haldane, Jaime~Fern{\'{a}}ndez
  del R{\'{i}}o, Mark Wiebe, Pearu Peterson, Pierre G{\'{e}}rard-Marchant,
  Kevin Sheppard, Tyler Reddy, Warren Weckesser, Hameer Abbasi, Christoph
  Gohlke, and Travis~E. Oliphant.
\newblock ``Array programming with {NumPy}''.
\newblock \href{https://dx.doi.org/10.1038/s41586-020-2649-2}{Nature {\bf 585},
  357--362}~(2020).

\bibitem{numba}
Siu~Kwan Lam, Antoine Pitrou, and Stanley Seibert.
\newblock ``Numba: a {LLVM}-based {Python} {JIT} compiler''.
\newblock In Proceedings of the {Second} {Workshop} on the {LLVM} {Compiler}
  {Infrastructure} in {HPC}.
\newblock \href{https://dx.doi.org/10.1145/2833157.2833162}{Pages 1--6}.
\newblock {LLVM} '15New York, NY, USA~(2015). Association for Computing
  Machinery.

\bibitem{gidney_stim_2021}
Craig Gidney.
\newblock ``Stim: a fast stabilizer circuit simulator''.
\newblock \href{https://dx.doi.org/10.22331/q-2021-07-06-497}{Quantum {\bf 5},
  497}~(2021).

\bibitem{hackersdelight}
Henry~S. Warren.
\newblock ``Hacker's delight''.
\newblock \href{https://dx.doi.org/10.5555/2462741}{Addison-Wesley
  Professional}. ~(2012).
\newblock 2nd edition.

\bibitem{kissinger2020Pyzx}
Aleks Kissinger and John van~de Wetering.
\newblock ``{PyZX: Large Scale Automated Diagrammatic Reasoning}''.
\newblock In Bob Coecke and Matthew Leifer, editors, {\rm Proceedings 16th
  International Conference on} Quantum Physics and Logic, {\rm Chapman
  University, Orange, CA, USA., 10-14 June 2019}.
\newblock \href{https://dx.doi.org/10.4204/EPTCS.318.14}{Volume 318 of
  Electronic Proceedings in Theoretical Computer Science, pages 229--241}.
\newblock Open Publishing Association~(2020).

\bibitem{amy_towards_2019}
Matthew Amy.
\newblock ``Towards {Large}-scale {Functional} {Verification} of {Universal}
  {Quantum} {Circuits}''.
\newblock \href{https://dx.doi.org/10.4204/EPTCS.287.1}{Electronic Proceedings
  in Theoretical Computer Science {\bf 287}, 1--21}~(2019).

\bibitem{yu_quantum_2021}
Nengkun Yu and Jens Palsberg.
\newblock ``Quantum abstract interpretation''.
\newblock In Proceedings of the 42nd {ACM} {SIGPLAN} {International}
  {Conference} on {Programming} {Language} {Design} and {Implementation}.
\newblock \href{https://dx.doi.org/10.1145/3453483.3454061}{Pages 542--558}.
\newblock {PLDI} 2021New York, NY, USA~(2021). Association for Computing
  Machinery.

\bibitem{mine_weakly_2004}
Antoine Min\'e.
\newblock ``Weakly {Relational} {Numerical} {Abstract} {Domains}''.
\newblock PhD Thesis~(2004).
\newblock  url:~\url{https://www-apr.lip6.fr/~mine/these/these-color.pdf}.

\bibitem{perdrix_quantum_2008}
Simon Perdrix.
\newblock ``Quantum {Entanglement} {Analysis} {Based} on {Abstract}
  {Interpretation}''.
\newblock In Proceedings of the 15th {International} {Symposium} on {Static}
  {Analysis}.
\newblock \href{https://dx.doi.org/10.1007/978-3-540-69166-2_18}{Pages
  270--282}.
\newblock {SAS} '08Berlin, Heidelberg~(2008). Springer-Verlag.

\bibitem{honda_analysis_2015}
Kentaro Honda.
\newblock ``Analysis of {Quantum} {Entanglement} in {Quantum} {Programs} using
  {Stabilizer} {Formalism}''.
\newblock \href{https://dx.doi.org/10.4204/EPTCS.195.19}{Electronic Proceedings
  in Theoretical Computer Science{\bf 195}}~(2015).

\bibitem{hietala-verif}
Kesha Hietala, Robert Rand, Shih-Han Hung, Liyi Li, and Michael Hicks.
\newblock ``{Proving Quantum Programs Correct}''.
\newblock \href{https://dx.doi.org/10.4230/LIPIcs.ITP.2021.21}{Leibniz
  International Proceedings in Informatics (LIPIcs) {\bf 193},
  21:1--21:19}~(2021).

\bibitem{chareton2020deductive}
Christophe Chareton, S{\'{e}}bastien Bardin, Fran{\c{c}}ois Bobot, Valentin
  Perrelle, and Beno{\^{\i}}t Valiron.
\newblock ``An automated deductive verification framework for circuit-building
  quantum programs''.
\newblock In Programming Languages and Systems.
\newblock \href{https://dx.doi.org/10.1007/978-3-030-72019-3_6}{Pages
  148--177}.
\newblock Springer International Publishing~(2021).

\bibitem{invariants_expo_2017}
Mingsheng Ying, Shenggang Ying, and Xiaodi Wu.
\newblock ``Invariants of quantum programs: Characterisations and generation''.
\newblock \href{https://dx.doi.org/10.1145/3093333.3009840}{SIGPLAN Not. {\bf
  52}, 818–832}~(2017).

\end{thebibliography}

%%%%%%%%%%%%%%%%%%%%
% BIBLIOGRAPHY END %
%%%%%%%%%%%%%%%%%%%%
% this message marks the end of the bibliography (used to split the paper from
% the supplement)
\message{^^JLASTREFERENCESPAGE \thepage^^J}

%%%%%%%%%%%%
% APPENDIX %
%%%%%%%%%%%%

% may decide to not include appendix% may decide to not include appendix
\ifincludeappendixx
	\clearpage
	\appendix
	% use \include to generate "appendix.aux". Needed for references into the
	% appendix
	\section{Abstract Transformers Soundness}\label{app:soundness}
Here, we prove the soundness of the trace transformer \cref{eq:def-abstract-trace}:
\begin{theorem}{Trace.}
	For all $\abe{\rho} \in \boldsymbol{\mathbb{D}}$ we have 
	$$\gamma \circ \ttrace{\abe{\rho}} \supseteq \trace \circ \gamma(\abe{\rho}).$$
\end{theorem}

\begin{proof}
	The over-approximation $\af{\grppref}$ follows closely the form of $\grppref$, where the first term $\pref(\abe{P})$ over-approximates the prefactors of $\abe{P}$ and second term over-approximates the prefactors originating from the solution space for $y$ of $\bare(\abe{P}) =\bare(\prod_{j = 1}^n Q_j^{y_j})$. 
	Overall, we have:
	\begin{small}
	\begin{align*}
		\trace \circ \gamma \left(\abe{\rho}\right) 
		&= 
		\trace \left(
			\left\{
			\sum_{i=1}^r
				c_i P_i \prod_{j=1}^n \twofrac
				\left(
					\I + (-1)^{b_{ij}} Q_j
				\right) \; \middle| \;
				% \right. \right.
				% \\
				% &\phantom{\spaceeq}\left.\left.
				c_i \in \abe{c}, P_i \in \abe{P}, b_{ij} \in \abe{b}_j
			\right\}
		\right)
		\\
		&= 
		\left\{
			\trace \left(
			\sum_{i=1}^r
				c_i P_i \prod_{j=1}^n \twofrac
				\left(
					\I + (-1)^{b_{ij}} Q_j
				\right)
				\right)
				\; \middle| \;
				% \right. \right.
				% \\
				% &\phantom{\spaceeq}\left.\left.
				c_i \in \abe{c}, P_i \in \abe{P}, b_{ij} \in \abe{b}_j
		\right\}
		\\
		&= 
		\left\{
			% \left(
			\sum_{i=1}^r
				\Re \left( 
					c_i \i^{\grppref(P, Q, b_{ij})}
				\right)
			% \right)
			\; \middle| \;
			c_i \in \abe{c}, P_i \in \abe{P}, b_{ij} \in \abe{b}_j
		\right\} & \text{Concrete trace, \cref{sec:trace}}
		\\
		&= 
			% \trace \left(
			\sum_{i=1}^r
				\Re \left( 
					\{c_i \in \abe{c}\} \cdot \i^{\grppref(\{P_i \in \abe{P}\}, Q, \{b_{ij} \in \abe{b}_j\})}
				\right)
			% \right)
		\\
		&\subseteq \gamma 
			\left(
			\sum_{i=1}^r
				\Re \left( 
					\abe{c} \cdot \i^{\grppref^\sharp(\abe{P}, Q, \abe{b}_j)}
				\right)
			\right) & \text{Soundness of transf.}
		\\
		&= \gamma 
			\left(
			r \cdot
				\Re \left( 
					\abe{c} \cdot \i^{\grppref^\sharp(\abe{P}, Q, \abe{b}_j)}
				\right)
			\right) & \text{Property of intervals}
		\\
		&= \gamma \circ 
			\trace \left( \abe{\rho} \right). 
	\end{align*}
	\end{small}
\end{proof}

\newpage

\section{Stabilizers and Pauli Matrices} \label{app:stabilizers}

\cref{tab:stabilizers} shows the states stabilized by each Pauli matrix,
together with the density matrix of the stabilized state.
Further, \cref{tab:multiply-pauli} shows the multiplication table for the Pauli
matrices.

\begin{table}
	\caption[States stabilized by Pauli matrices]{States stabilized by Pauli
	matrices $\pauli$ and also $-\pauli$, where $\X := \Xmatrix$, $\Y :=
	\Ymatrix$, $\Z := \Zmatrix$, and $\I[2] := \Imatrix$.}
	\label{tab:stabilizers}
	\small
	\centering
	\newcommand{\myfrac}[2]{\tfrac{#1}{#2}}
	\setlength{\tabcolsep}{1pt}
	\def\arraystretch{1.4}
	\begin{tabular}{@{}ccc@{\hskip 2.5pt}|@{\hskip 2.5pt}ccc@{}}
		\toprule
		\textbf{Stab.} & \textbf{State vec.} & \textbf{Dens. mat.} &
		\textbf{Stab.} & \textbf{State vec.} & \textbf{Dens. mat.} \\
		\midrule
		% X
		$\X$ &
		$\myfrac{1}{\sqrt{2}}\begin{psmallmatrix}
			1 \\
			1
		\end{psmallmatrix} \equalhat \ket{+}$ &
		$\myfrac{1}{2}\begin{psmallmatrix}
			\phantom{-}1 & 1 \\
			\phantom{-}1 & 1
		\end{psmallmatrix}$ & 
		% -X
		$-\X$ &
		$\myfrac{1}{\sqrt{2}}\begin{psmallmatrix}
			\phantom{-}1 \\
			-1
		\end{psmallmatrix} \equalhat \ket{-}$ &
		$\myfrac{1}{2} \begin{psmallmatrix}
			\phantom{-}1 & -1 \\
			-1 & \phantom{-}1
		\end{psmallmatrix}$ \\
		% Y
		$\Y$ &
		$\myfrac{1}{\sqrt{2}}\begin{psmallmatrix}
			1 \\
			\i
		\end{psmallmatrix} \phantom{\hspace{3pt}\equalhat \ket{+}}$ &
		$\myfrac{1}{2}\begin{psmallmatrix}
			\phantom{-}1 & \i \\
			-\i & 1
		\end{psmallmatrix}$ & 
		% -Y
		$-\Y$ &
		$\myfrac{1}{\sqrt{2}}\begin{psmallmatrix}
			\phantom{-}1 \\
			-\i
		\end{psmallmatrix} \phantom{\hspace{3pt}\equalhat \ket{+}}$ &
		$\myfrac{1}{2}\begin{psmallmatrix}
			\phantom{-}1 & -\i \\
			\phantom{-}\i & \phantom{-}1
		\end{psmallmatrix}$ \\
		% Z
		$\Z$ &
		$\phantom{\myfrac{1}{\sqrt{2}}}\hspace{-2pt}\begin{psmallmatrix}
			1 \\
			0
		\end{psmallmatrix}
		\equalhat \ket{0}$ &
		$\phantom{\myfrac{1}{2}}\begin{psmallmatrix}
			\phantom{-}1 & 0 \\
			\phantom{-}0 & 0
		\end{psmallmatrix}$ & 
		% -Z
		$-\Z$ &
		$\phantom{\myfrac{1}{\sqrt{2}}}\hspace{-2pt}\begin{psmallmatrix}
			\phantom{-}0 \\
			\phantom{-}1
		\end{psmallmatrix}
		\equalhat \ket{1}$ &
		$\phantom{\myfrac{1}{2}}\begin{psmallmatrix}
			\phantom{-}0 & \phantom{-}0 \\
			\phantom{-}0 & \phantom{-}1
		\end{psmallmatrix}$ \\
		% I
		$\I[2]$ &
		(any vec.) &
		- &
		$-\I[2]$ &
		(no vec.) &
		- \\
		\bottomrule
	\end{tabular}
\end{table}

\begin{table}
	\caption{Multiplication of Pauli matrices.}
	\label{tab:multiply-pauli}
	\centering
	\small
	\setlength{\tabcolsep}{1.5pt}
	\begin{tabular}
		{
			@{}
			*{3}
			{
				ccr@{\hspace{20pt}}
			}
			ccr@{}
		}
		\toprule
		$\I \I$ & $=$ & $    \I$ &
		$\I \X$ & $=$ & $    \X$ &
		$\I \Y$ & $=$ & $    \Y$ &
		$\I \Z$ & $=$ & $    \Z$
		\\
		$\X \I$ & $=$ & $    \X$ &
		$\X \X$ & $=$ & $    \I$ &
		$\X \Y$ & $=$ & $ \i \Z$ &
		$\X \Z$ & $=$ & $-\i \Y$ \\
		$\Y \I$ & $=$ & $    \Y$ &
		$\Y \X$ & $=$ & $-\i \Z$ &
		$\Y \Y$ & $=$ & $    \I$ &
		$\Y \Z$ & $=$ & $ \i \X$ \\
		$\Z \I$ & $=$ & $    \Z$ &
		$\Z \X$ & $=$ & $ \i \Y$ &
		$\Z \Y$ & $=$ & $-\i \X$ &
		$\Z \Z$ & $=$ & $    \I$ \\
		\bottomrule
	\end{tabular}
\end{table}

\fi

%%%%%%%
% END %
%%%%%%%
% this message marks the end of the document (important for splitting of paper)
\message{^^JLASTPAGE \thepage^^J}

\end{document}